\documentclass[preprint,12pt]{elsarticle}

\usepackage{amsmath,amsfonts,amsthm,amssymb,paralist,subfigure,graphicx,amsbsy,float,epsfig,color}
\usepackage{cuted,mathtools,lipsum}

\usepackage[ruled,vlined,linesnumbered]{algorithm2e}
\usepackage{stmaryrd,url}

\usepackage{amsthm}
\newtheorem{lem}{Lemma}
\newtheorem{ass}{Assumption}
\newtheorem{theorem}{Theorem}
\newtheorem{defn}{Definition}
\newtheorem{rem}{Remark}
\newtheorem{prop}{Proposition}
\newtheorem{cor}{Corollary}

\def\mb{\mathbf}

\def\mc{\mathcal}

\def\mb{\mathbf}

\def\mc{\mathcal}

\journal{European Journal of Control}

\begin{document}

\begin{frontmatter}

\title{ Distributed Allocation and Resource Scheduling Algorithms Resilient to Link Failure
}

\author[Sem]{Mohammadreza Doostmohammadian}
\affiliation[Sem]{Mechatronics Group, Faculty of Mechanical Engineering, Semnan University, Semnan, Iran, doost@semnan.ac.ir.}
\author[SP]{Sergio Pequito}
\affiliation[SP]{Department of Electrical and Computer Engineering and Institute for Systems and Robotics, Instituto Superior Tecnico, University of Lisbon, Portugal, sergio.pequito@tecnico.ulisboa.pt.}

\begin{abstract}
Distributed resource allocation (DRA) is fundamental to modern networked systems, spanning applications from economic dispatch in smart grids to CPU scheduling in data centers. Conventional DRA approaches require reliable communication, yet real-world networks frequently suffer from link failures, packet drops, and communication delays due to environmental conditions, network congestion, and security threats.

We introduce a novel resilient DRA algorithm that addresses these critical challenges, and our main contributions are as follows: (1) guaranteed constraint feasibility at all times, ensuring resource-demand balance even during algorithm termination or network disruption; (2) robust convergence despite sector-bound nonlinearities at nodes/links, accommodating practical constraints like quantization and saturation; and (3) optimal performance under merely uniformly-connected networks, eliminating the need for continuous connectivity.

Unlike existing approaches that require persistent network connectivity and provide only asymptotic feasibility, our graph-theoretic solution leverages network percolation theory to maintain performance during intermittent disconnections. This makes it particularly valuable for mobile multi-agent systems where nodes frequently move out of communication range. Theoretical analysis and simulations demonstrate that our algorithm converges to optimal solutions despite heterogeneous time delays and substantial link failures, significantly advancing the reliability of distributed resource allocation in practical network environments.
\end{abstract}

%%Graphical abstract
\begin{graphicalabstract}
	\includegraphics{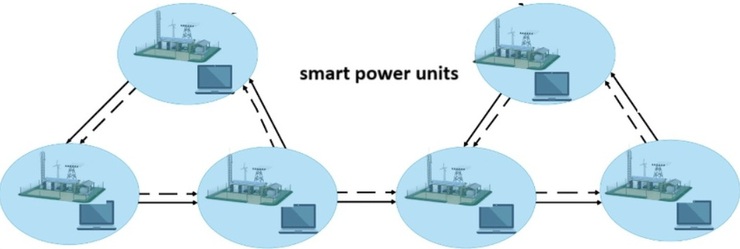}
\end{graphicalabstract}

%%Research highlights
\begin{highlights}
	\item Clearly defining the problem of distributed resource allocation as a formal mathematical optimization setup.
	\item Proposing a distributed resource allocation algorithm with optimal convergence over uniformly-connected networks while addressing all-time feasibility.
	\item Proving algorithm convergence and resilience under random link failure and time-delayed communication links
\end{highlights}


\end{frontmatter}

\section{Introduction} \label{sec_intro}

Modern networked systems increasingly rely on distributed algorithms for resource allocation, yet these systems face significant reliability challenges when operating in environments where communication disruptions are common. Resilient Distributed Resource Allocation (DRA) has emerged as a critical research area focused on maintaining optimal performance despite network imperfections such as link failures, packet drops, time delays, and cyber-attacks \cite{pirani2023graph,ARC}.
Resilient DRA algorithms serve diverse applications, from CPU scheduling that optimally allocates computation loads among networked data centers \cite{cpu,oliva2022best}; coordinated coverage control distributes surveillance tasks among mobile sensors to minimize distance-based costs \cite{MSC09,telsang2022decentralized}; Task allocation for maximizing reliability of distributed computing systems \cite{kang2010task,yin2007task}; distributed energy management assigns power demand across generators to optimize generation costs \cite{zhou2022distributed2,cherukuri2016initialization}; and electric vehicle charging coordination balances grid stability while meeting charging needs \cite{shah2019optimization,falsone2020tracking}.

Despite advances in resilient algorithms across various distributed domains, a critical gap remains in resource allocation systems that can maintain both constraint feasibility and optimality under realistic network conditions. Current DRA approaches -- whether linear \cite{boyd2006optimal,zhang2020distributed,doan2017ccta,rogozin2022decentralized}, predefined-time \cite{lin2020predefined,wang2024predefined}, attack-resilient \cite{SHAO20206241}, or event-triggered \cite{cai2023event} -- largely overlook the practical challenges of communication latency and nonlinear agent dynamics. Even approaches that address directed network topologies \cite{zhou2022distributed,9398581,zhu2019distributed} remain vulnerable to link failures and depend on continuous network connectivity, an unrealistic assumption in many practical settings.

Communication disruptions are pervasive in real-world multi-agent systems through packet retransmission requirements, processing delays, and physical limitations. When DRA algorithms fail to address these realities, they risk violating resource-demand constraints and converging to suboptimal solutions, potentially triggering cascading instabilities or complete service disruptions. While consensus algorithms have made strides in addressing communication delays \cite{hadjicostis2013average}, existing delay-tolerant DRA approaches \cite{wang2019distributed,cdc_dtac} achieve constraint satisfaction only asymptotically and still depend on persistent network connectivity -- assumptions that break down in environments with intermittent communication or mobile agents.

To bridge this critical gap, we introduce a graph-theoretic framework specifically designed to handle link failures and packet drops in dynamic network environments. Our approach recognizes that real-world networks frequently experience topology changes due to long-distance communication challenges, environmental interference, and node mobility. Unlike conventional algorithms that require continuous connectivity, our solution tolerates networks that are merely uniformly connected, meaning they may temporarily fragment but regain connectivity within bounded time intervals.

Furthermore, we develop a delay-tolerant mechanism that maintains feasibility and optimal convergence despite heterogeneous time delays across communication links. Most significantly, our resilient solution guarantees all-time feasibility -- maintaining perfect balance between allocated resources and demand at every iteration. This critical property ensures that the algorithm can be safely terminated at any point without risking constraint violations, even during network disruptions or recovery phases.

In summary, our main contributions are as follows:

(i) We introduce a distributed resource allocation algorithm that operates over uniformly-connected networks, significantly advancing resilience to link failures, packet drops, and time delays without requiring continuous network connectivity;

(ii) Our algorithm achieves exact convergence to the optimizer despite sector-bound nonlinear mappings at both links and nodes, accommodating practical constraints such as log-scale quantization and clipping, a capability not previously addressed in distributed resource allocation literature;

(iii) The proposed approach maintains perfect balance between allocated resources and demand at all iterations (all-time feasibility), even with nonlinear constraints on nodes/links. This critical property enables safe algorithm termination at any point without risking resource-demand imbalance or service disruption, a feature absent in dual-based formulations and ADMM-based methods \cite{banjac2019decentralized,jiang2022distributed,wang2020dual,Carli_ADMM,falsone2020tracking}; and

(iv) Our solution guarantees optimality, all-time feasibility, and convergence despite heterogeneous time delays and packet drops across the network. This comprehensive resilience ensures the algorithm functions effectively even when some links experience significant delays or complete failure, advancing beyond existing delay-tolerant approaches \cite{wang2019distributed,cdc_dtac} that cannot handle sector-bound nonlinearity while maintaining resource-demand feasibility.

The rest of the paper is organized as follows. Section~\ref{sec_prob} states the DRA problem and its applications. Section~\ref{sec_pre} presents the preliminaries on network science and graph theory. Section~\ref{sec_alg} provides our main distributed algorithm to solve the DRA problem for both delay-free and delayed cases. Section~\ref{sec_conv} discusses the general convergence analysis, also in the presence of link failure and time delay. Section~\ref{sec_sim} presents the simulations and Section~\ref{sec_con} concludes the paper.

\section{The Resilient DRA Problem} \label{sec_prob}

Consider a multi-agent network where agents (i.e., nodes) have some (possibly nonlinear) dynamics and have to exchange information to find the solution to an optimization problem where some constraints need to be satisfied across the multi-agent network while the exchanged information is local. The network of agents is modelled as a graph of $n$ nodes, denoted by $\mathcal{G} = \{\mathcal{V},\mathcal{E}\}$ that might be time-varying, i.e., for a fixed set of nodes $\mathcal{V}$ the link set $\mathcal{E}$ may change over time. Its associated adjacency matrix, denoted by $W$, is symmetric and entry $W_{ij} > 0$ denotes the weighting factor on the information shared over link $(j, i) \in \mathcal{E}$. $W_{ij}=0$ implies that $(j, i) \notin \mathcal{E}$. Define $\mathcal{N}_i = \{j | (j,i) \in \mathcal{E}\}$ as the set of neighbours of the node/agent $i$ for data-sharing.

The DRA problem in its most general representation is in the form of coupling-constraint distributed optimization. The objective function represents the sum of local cost (or loss) at agents, each having access to a certain amount of resources $z_i$. The weighted sum of the shared resources is constant~$b$. The problem is to minimize the objective by assigning the proper amount of resources to the agents subject to the resource-demand equality constraint as follows
\begin{align} \label{eq_dra0}
\min_{\mathbf{z}}
~~ & H(\mathbf{z}) = \sum_{i=1}^{n} h_i(z_i),\\ \nonumber
&\text{s.t.} ~  \sum_{i=1}^n a_iz_i  = b.
\end{align}

By a simple change of variable $a_iz_i=x_i$, the problem changes into the following standard form
\begin{align} \label{eq_dra}
\min_{\mathbf{x}}
~~ & F(\mathbf{x}) = \sum_{i=1}^{n} f_i(x_i),\\ \nonumber
&\text{s.t.} ~  \sum_{i=1}^n x_i  = b.
\end{align}

The function $f_i: \mathbb{R} \mapsto \mathbb{R}$ is strictly convex and $2u$-smooth, i.e., it satisfies $0<\frac{d^2f_i}{dx_i^2} < 2u$. In many applications, the constraint $\sum_{i=1}^n x_i  = b$ is required to be held at all iterations under the distributed solution. This is referred to as \emph{all-time feasibility} and ensures that the balance between allocated states $\sum_{i=1}^n x_i$ and demand $b$ holds at any termination time of the algorithm with no service disruption.

The states of the agents might be subject to local box constraints which represent simple upper-bound and lower-bound on the states in the form
\begin{align} \label{eq_box}
	\min_{\mathbf{x}}
	~~ & F(\mathbf{x}) = \sum_{i=1}^{n} f_i(x_i),\\ \nonumber
	&\text{s.t.} ~  \sum_{i=1}^n x_i  = b, \\ \nonumber
	&m_i \leq x_i \leq M_i.
\end{align}

The local constraints in \eqref{eq_box} can be efficiently addressed by incorporating barrier functions or additive penalty terms into the objective function. This approach introduces a controllable optimality gap as discussed in \cite{nesterov1998introductory,bertsekas2003convex}.

A common penalty formulation is $f^\sigma(x_i) = \sigma ([x_i-M_i]^+ + [m_i-x_i]^+)$, where $[u]^+ = (\max\{0,u\})^m$ with $m \in \mathbb{N}_{\geq 2}$. The parameters $m$ and $\sigma$ should be sufficiently large to ensure adequate penalization of constraint violations, thereby guiding the solution toward the feasible region.

For applications requiring smoother penalty transitions, the logarithmic approximation ${f^\mu(u)=\frac{1}{\mu}\log (1+\exp(\mu u))}$ offers an attractive alternative, where $u$ represents either $[x_i-M_i]$ or $[m_i-x_i]$. This formulation converges to $\max\{u,0\}$ as $\mu$ increases, providing arbitrarily close approximations to the exact constraint with appropriate parameter selection. See \cite{nesterov1998introductory} for a comprehensive analysis of these penalty methods and their convergence properties.

To facilitate reading and ease the exposition that follows, all required notation is summarized in Table~\ref{tab_notation}.

\begin{table} [hbpt!]
	\centering
		\caption{List of symbols}\
		\scriptsize\setlength{\tabcolsep}{3pt}
		\begin{tabular}{|c|c|}
			\hline 		\hline
			$\mb{x}$, $\mb{z}$ & column state vectors \\ \hline
			$\mb{1}_n$, $\mb{0}_n$ & column vector of ones and zeros of size $n$ \\
			\hline
			$\mb{x}^*$ & optimal state \\ \hline
			${\mc{S}}_b$ & feasibility set (associated with $b$) \\
			\hline
			$x_i$ & state of node $i$ \\\hline
			$k$ & discrete time index \\
			\hline
			$F(\cdot)$ & global cost function \\
			\hline $\mc{G}$, $\mc{V}$, $\mc{E}$ & network, set of nodes, set of links \\
			\hline
			$\nabla F(\cdot)$  & gradient of $F(\cdot)$ \\ \hline
			$\overline{d_G}$ & network diameter \\
			\hline
			$f_i(\cdot)$ & local cost function  \\
			\hline $W$, $L$ & weight matrix, Laplacian matrix \\
			\hline
			$\partial f_i(\cdot)$ & first derivative of $f_i(\cdot)$  \\  \hline $\lambda_i$ & $i$th eigenvalue of $L$ \\
			\hline
			$\tau_{ij}$ & time-delay of link $(i,j)$ \\		
			\hline
			$\mc{N}_i$ & neighbors of node $i$ \\ \hline  $\overline{\tau}$ & global upper-bound on the time-delay \\
			\hline
			$g_n,g_l$ & nonlinear mapping at the nodes/links \\ \hline
			$\kappa_n,\kappa_l$ & lower sector-bound on $g_n,g_l$ \\ \hline	
			$\mc{K}_n,\mc{K}_l$ & upper sector-bound on $g_n,g_l$ \\ \hline
			$u$ & upper bound on second gradient of $f_i$ \\ \hline	
			$\mc{I}(\cdot)$ & indicator function \\
			\hline		
			$n$ & number of agents/nodes \\ \hline $\zeta $   & dispersion parameter vector \\
			\hline $p_c,p_l$   & percolation threshold, link failure probability \\
			\hline		
			\hline
		\end{tabular} \normalsize
		\label{tab_notation}
\end{table}

\subsection{Resilient DRA Problem Formulation}

Hereafter, we address the challenge of distributed resource allocation over unreliable networks characterized by communication disruptions. Unlike conventional approaches that assume reliable communication, we explicitly model and mitigate the effects of random packet drops and link failures that occur in practical networked systems.

Our focus on unreliable networks is motivated by several key considerations:

\begin{enumerate}
\item \textbf{Random Structure of Real Networks}: Many real-world networks exhibit inherently random characteristics \cite{barabasi}. By modelling networks as random structures, we can better analyze their interconnection properties and behaviour under varying conditions.

\item \textbf{Unpredictable Link Failures}: In practical deployments, link failures occur unpredictably due to hardware malfunctions, environmental interference, or congestion. Our random network model allows rigorous evaluation of algorithm resilience against such failures.

\item \textbf{Scalability Analysis}: Random network models facilitate scalability analysis as system complexity increases, providing insights into how distributed algorithms perform with expanding network size and connectivity changes.
\end{enumerate}

Our resilient DRA formulation extends the standard problem to incorporate:

\begin{enumerate}
\item \textbf{Time-Varying Connectivity}: We model $\mathcal{G}(k) = \{\mathcal{V},\mathcal{E}(k)\}$ as a dynamic graph where the edge set $\mathcal{E}(k)$ changes over time due to communication disruptions.

\item \textbf{Probabilistic Link Failures}: Each communication link $(i,j) \in \mathcal{E}$ has a probability $p_l$ of failure at any time step, representing packet drops, congestion, or physical disconnection.

\item \textbf{Nonlinear Communication Effects}: We account for practical nonlinearities in information exchange through sector-bound mappings at both nodes and links $g_n: \mathbb{R} \mapsto \mathbb{R}$, and  $g_l: \mathbb{R} \mapsto \mathbb{R}$,
where $\kappa_n z < g_n(z) < K_n z$ and $\kappa_l z < g_l(z) < K_l z$ represent the sector bounds on these nonlinear mappings.

\item \textbf{Heterogeneous Time Delays}: Communication between agents may experience varying delays $\tau_{ij}(k) \leq \bar{\tau}$ that are time-varying, random, and potentially different for each link. These delays are incorporated into our agent dynamics as
{\footnotesize
\begin{align}
\hspace{-1cm}x_i(k+1) = x_i(k) - \eta_{\tau} \sum_{j \in \mathcal{N}_i} \sum_{r=0}^{\bar{\tau}} W_{ij} g_n\left(g_l(\partial f_i(k-r)) - g_l(\partial f_j(k-r))\right) I_{k-r,ij}(r),
\end{align}
}
where $I_{k-r,ij}(r)$ is an indicator function defined as
\begin{align}
I_{k,ij}(\tau) =
\begin{cases}
1, & \text{if } \tau_{ij}(k) = \tau, \\
0, & \text{otherwise}.
\end{cases}
\end{align}
\end{enumerate}

This comprehensive model captures the realistic challenges of distributed resource allocation in practice, where network disruptions can occur due to hardware malfunctions, interference, congestion, or buffer overflow. Our solution approach, based on the percolation theory discussed in Section~\ref{sec_perc}, provides rigorous guarantees for algorithm performance despite these communication challenges.

%Using Karush-Kuhn-Tucker (KKT) conditions one can prove that the optimal point $\mb{x}^*$ of problem~\eqref{eq_dra} satisfies $\nabla F(\mb{x}^*) \in \mbox{span}(\mb{1}_n)$ (with $\mb{1}_n$ as vector of all ones).

\section{Preliminaries on Network Science} \label{sec_pre}

This section establishes the mathematical and theoretical foundations necessary for our resilient distributed resource allocation approach. We begin by introducing essential concepts from algebraic graph theory in Section~\ref{subsec:agt}, where we examine the properties of Laplacian matrices and their eigenspectra.  Next, in Section~\ref{sec_perc}, we introduce network percolation theory, which provides the analytical framework for characterizing network resilience under random link failures. These mathematical preliminaries form the essential basis for both our algorithm design in Section~\ref{sec_alg} and the rigorous convergence analysis presented in Section~\ref{sec_conv}, particularly when proving algorithm resilience to link failures and communication delays.

\subsection{Algebraic Graph Theory}\label{subsec:agt}
Next, we present essential graph-theoretic concepts relevant to our analysis. Given a multi-agent network with adjacency matrix $W$, we define the diagonal degree matrix $D := \text{diag}(\sum_{j=1}^n W_{ij})$ and the corresponding Laplacian matrix $L = D - W$. For a connected graph with symmetric $L$, the spectrum of the Laplacian contains exactly one zero eigenvalue, with $\mathbf{1}_n$ (the all-ones vector) as its associated eigenvector. When the network becomes disconnected or fragmented into islands, the Laplacian matrix exhibits multiple zero eigenvalues, with the number corresponding to the count of connected components. The spectral properties of the Laplacian matrix $L$ fundamentally characterize the consensus behaviour and information diffusion across the network $\mathcal G$. The following lemma, derived from the Courant-Fischer theorem, provides crucial insight into quantifying consensus \emph{disagreement} in multi-agent networks~\cite{SensNets:Olfati04}.

\begin{lem} \label{lem_xLy}
\cite{SensNets:Olfati04} Given a symmetric Laplacian $L$ and vector $\mb{x} \in \mathbb{R}^n$, define $\overline{\mb{x}} =: \mb{x} - \frac{\mb{1}_n^\top \mb{x}}{n} \mb{1}_n$ as the dispersion state.
Then, the following hold:
\begin{align} \label{eq_laplace1}
	\mb{x}^\top L \mb{x} &= \overline{\mb{x}}^\top L \overline{\mb{x}},
	%\mb{x}^\top L \mb{y} = \overline{\mb{x}}^\top L \overline{\mb{y}},
	% \\ \label{eq_laplace2}
	%  &\mb{x}^\top L \mb{y} <  \mb{y}^\top L^2 \mb{y} ~~\mbox{\mb{if}}~~  \overline{\mb{x}}^\top \overline{\mb{x}} < \overline{\mb{y}}^\top L^2 \overline{\mb{y}}
	\\      \label{eq_laplace}
	\lambda_2 \|\overline{\mb{x}} \|_2^2 \leq \mb{x}^\top &L\mb{x} \leq \lambda_n \|\overline{\mb{x}} \|_2^2,
\end{align}
where $\lambda_n$ and $\lambda_2$ respectively denote the largest and smallest non-zero eigenvalue of $L$.
\end{lem}

One can extend the results of Lemma~\ref{lem_xLy} to consider sector-bound nonlinearity at the nodes/links as follows.

\begin{cor} \label{cor_xLy}
	 Consider two vector states $\mb{x},\mb{y} \in \mathbb{R}^n$ where the entries $y_i=g(x_i)$ for $i \in \{1,\dots,n\}$ such that the sector-bound mapping $g(\cdot)$ is sign-preserving, odd, and monotonically non-decreasing implying $0<\kappa <\frac{y_i}{x_i} < \mc{K}$, where $\kappa, \mc{K} \in \mathbb{R}^+$ are lower/upper sector bound. Let $\overline{\mb{x}} =: \mb{x} - \frac{\mb{1}_n^\top \mb{x}}{n} \mb{1}_n$ and $\overline{\mb{y}} =: \mb{y} - \frac{\mb{1}_n^\top \mb{y}}{n} \mb{1}_n$. Then, the following holds:
	\begin{align} \label{eq_laplace1_xy}
		\mb{x}^\top L \mb{y} &= \overline{\mb{x}}^\top L \overline{\mb{y}},
		\\      \label{eq_laplace_xy}
		\lambda_2 \overline{\mb{x}}^\top \overline{\mb{y}} \leq \mb{x}^\top &L\mb{y} \leq \lambda_n \overline{\mb{x}}^\top \overline{\mb{y}}.
	\end{align}
\end{cor}
\begin{proof}
	Recall that for symmetric $W$ matrix, its associated Laplacian $L$ satisfies $\mb{1}^\top_n  L = L \mb{1}_n = \mb{0}_n$. Then,
	\begin{align} \nonumber
		\mb{x}^\top L \mb{y} &= (\overline{\mb{x}} + \frac{\mb{1}_n^\top \mb{x}}{n} \mb{1}_n)^\top L (\overline{\mb{y}} + \frac{\mb{1}_n^\top \mb{y}}{n} \mb{1}_n) \\ \nonumber
		&=  \overline{\mb{x}}^\top L \overline{\mb{y}},
	\end{align}
	which proves \eqref{eq_laplace1_xy}.
	The proof of the second part follows from a special case of Courant-Fischer's theorem in \cite{hornjohnson}. Denote the eigenvalues of $L$ as $\lambda_1 <\lambda_2\leq \lambda_3 \leq \dots \leq \lambda_n$ with $\lambda_1 = 0$ associated with the eigenvalue $\mb{1}_n$. We can express $\mb{x},\mb{y}$ in terms of the eigenbasis of $L$, denoted by unit vectors $v_1,\dots,v_n$, as
	\begin{align} \nonumber
		\mb{x} = \alpha_1 v_1+ \dots + \alpha_n v_n, ~ \mbox{and}~ \mb{y} = \beta_1 v_1+ \dots + \beta_n v_n,
	\end{align}	
	where $0 < \kappa <\frac{\alpha_i}{\beta_i} < \mc{K}$. Expanding $\mb{x}^\top L \mb{y} = \overline{\mb{x}}^\top L \overline{\mb{y}}$, and  using the eigendecomposition of $L$, we get
	\begin{align}\nonumber
		\overline{\mb{x}}^\top L \overline{\mb{y}} = \alpha_1 \beta_1 \lambda_1 +\alpha_2 \beta_2 \lambda_2 + \dots + \alpha_n \beta_n \lambda_n,
	\end{align}
	where due to sector-bound condition and monotonicity, we have $0<\kappa \beta_i^2 < \alpha_i \beta_i < \mc{K} \beta_i^2$. We also have
	$\overline{\mb{x}}^\top \overline{\mb{y}} = \alpha_1 \beta_1 + \dots + \alpha_n \beta_n.$
	Then, following the non-decreasing order of eigenvalues and the fact that $\lambda_1 =0$ is associated with the eigenvector $\mb{1}_n$, we have
	\begin{align}\nonumber
		\lambda_2 = \min_{\overline{\mb{x}},\overline{\mb{y}} \neq 0, \overline{\mb{x}} \perp \mb{1}_n,\overline{\mb{y}}  \perp \mb{1}_n} \frac{\overline{\mb{x}}^\top L \overline{\mb{y}}}{\overline{\mb{x}}^\top \overline{\mb{y}}}, ~\mbox{and}~ \lambda_n = \max_{\overline{\mb{x}}, \overline{\mb{y}} \neq 0, \overline{\mb{x}} \perp \mb{1}_n,\overline{\mb{y}}  \perp \mb{1}_n} \frac{\overline{\mb{x}}^\top L \overline{\mb{y}}}{\overline{\mb{x}}^\top \overline{\mb{y}}}
	\end{align}
	that is known as the Rayleigh quotient \cite{hornjohnson}, from which it readily follows~\eqref{eq_laplace_xy}.
\end{proof}
\begin{lem}  \label{lem_xLy2}
	Let $\mb{y}= g(\mb{x})$ such that the sector-bound mapping $g(\cdot)$ is sign-preserving, odd, and monotonically non-decreasing such that the element-wise relations $ \kappa x_i < y_i < \mc{K} x_i$ hold for all $i,j$. Then, it follows that
	\begin{align} \label{eq_laplace2}
		\lambda_2 \kappa \|\overline{\mb{x}} \|_2^2\leq \mb{x}^\top L &\mb{y} \leq \lambda_n \mc{K} \|\overline{\mb{x}} \|_2^2.
	\end{align}
\end{lem}
\begin{proof}
	From the definition of $\overline{\mb{x}}, \overline{\mb{y}}$, we have
	\begin{align} \nonumber
		\overline{\mb{x}}^\top \overline{\mb{y}} &= (\mb{x} - \frac{\mb{1}_n^\top \mb{x}}{n} \mb{1}_n)^\top(\mb{y} - \frac{\mb{1}_n^\top \mb{y}}{n} \mb{1}_n)\\ \nonumber
		&= \mb{x}^\top \mb{y}  - \frac{\mb{1}_n^\top \mb{x}}{n} \mb{1}_n^\top\mb{y} - \frac{\mb{1}_n^\top \mb{y}}{n} \mb{1}_n^\top \mb{x} + \frac{\mb{1}_n^\top \mb{y} \mb{1}_n^\top \mb{x}}{n} \mb{1}_n^\top \mb{1}_n
		\\ \label{eq_proof_xyn}
		&= \mb{x}^\top \mb{y} + \frac{\mb{1}_n^\top \mb{y} \mb{1}_n^\top \mb{x}}{n} (n-2).
	\end{align}	
	From the sector-bound relation, we have $ \kappa \mb{x} < \mb{y} < \mc{K} \mb{x}$ and applying this to~\eqref{eq_proof_xyn}, we obtain
	\begin{align} \nonumber
		\kappa (\mb{x}^\top \mb{x} + \frac{\mb{1}_n^\top \mb{x} \mb{1}_n^\top \mb{x}}{n} (n-2)) &\leq \mb{x}^\top \mb{y} + \frac{\mb{1}_n^\top \mb{y} \mb{1}_n^\top \mb{x}}{n} (n-2) \\ \label{eq_proof_xyn1}
		&\leq \mc{K} (\mb{x}^\top \mb{x} + \frac{\mb{1}_n^\top \mb{x} \mb{1}_n^\top \mb{x}}{n} (n-2)).
	\end{align}
	Similar to \eqref{eq_proof_xyn}, from the definition of $\overline{\mb{x}}$, it readily follows that
	\begin{align} \label{eq_proof_xyn2}
		\mb{x}^\top \mb{x} + \frac{\mb{1}_n^\top \mb{x} \mb{1}_n^\top \mb{x}}{n} (n-2) = \overline{\mb{x}}^\top \overline{\mb{x}} = \|\overline{\mb{x}} \|_2^2.
	\end{align}
	Combining Eqs.~\eqref{eq_proof_xyn},~\eqref{eq_proof_xyn1}, and ~\eqref{eq_proof_xyn2}, we get  $\kappa \|\overline{\mb{x}} \|_2^2 \leq \overline{\mb{x}}^\top \overline{\mb{y}}\leq \mc{K} \|\overline{\mb{x}} \|_2^2$. Then, combining this with Eq.~\eqref{eq_laplace_xy} in Corollary~\ref{cor_xLy}, we obtain~\eqref{eq_laplace2}.
\end{proof}

Note that the second smallest eigenvalue of $L$, denoted by  $\lambda_2$, is known as algebraic connectivity. It is known from the literature that adding links to the network $\mc{G}$ increases its algebraic connectivity \cite{SensNets:Olfati04}.  These results can be extended to the union network as follows. First, let us recall the definition of the uniformly-connectivity.
\begin{defn}\label{def1}
	For a uniformly-connected network, there exists \textit{bounded} time $T>0$ such that $\mc{G}_T(k) = \cup_{k}^{k+T} \mc{G}(k)$ is connected for all $k$.
\end{defn}

\begin{rem}
The boundedness of $T$ in Definition~\ref{def1} is essential for uniform-connectivity, as it guarantees the network regains connectivity within a finite time window. Consider a counterexample: a network that connects only during sporadic intervals and remains disconnected throughout most of the time domain. In such a case, no bounded $T$ exists for which $\mathcal{G}_T(k) = \cup_{j=k}^{k+T} \mathcal{G}(j)$ is connected for all $k \geq 0$. Therefore, this network fails to satisfy the uniform-connectivity property. Uniform connectivity requires not just eventual reconnection, but reconnection within a consistent, bounded timeframe across the entire time domain. \hfill $\circ$
\end{rem}

One can extend the results of Lemma~\ref{lem_xLy} and Corollary~\ref{cor_xLy} to uniformly-connected networks as follows.

\begin{cor} \label{cor_xLy_union}
 Consider uniformly-connected network $\mc{G}_T(k) = \cup_{k}^{k+T} \mc{G}(k)$ with associated Laplacian $L_T$ and eigenvalues $\lambda_{2T}$ and $\lambda_{nT}$. Then,  $\lambda_{2T} \overline{\mb{x}}^\top \overline{\mb{y}} \leq \mb{x}^\top L_T\mb{y} \leq \lambda_{nT} \overline{\mb{x}}^\top \overline{\mb{y}}$ and, following from Lemma~\ref{lem_xLy2}, we have  $\lambda_{2T}  \kappa \|\overline{\mb{x}} \|_2^2\leq \mb{x}^\top L_T \mb{y} \leq \lambda_{nT} \mc{K} \|\overline{\mb{x}} \|_2^2$.
\end{cor}

%Similar statements as in Corollary~\ref{cor_xLy_union} hold for Lemma~\ref{lem_xLy2}. {\color{red} not clear what is your intention here... if there is something you would like to be stated, it is better to do so... there is space}

\subsection{Percolation Theory} \label{sec_perc}

Percolation theory is a mathematical framework that studies the behaviour of connected clusters in random graphs or networks.  Specifically, the theory investigates how global connectivity emerges from local random processes, focusing on phase transitions where the system's properties change dramatically at critical thresholds. The key challenge percolation theory addresses is determining when a network maintains its essential connectivity properties despite random removal of links or nodes\cite{cohen2000resilience,li2021percolation}.  In many control applications, the connectivity of the graph plays a key role in the convergence. This is addressed by the notion of \emph{bond-percolation} as defined next.

\begin{defn} \cite{sahimi1994applications}
Let denote the probability of link removal by $p_l$. Then,  \textit{bond-percolation} of the network refers to the probability threshold $p_c$ (or the likelihood) that a network $\mc{G}$  will lose its connectivity when links are randomly removed or \textit{percolated} from $\mc{G}$. In other words, for $p_l>p_c$ there is no giant connected component spanning all nodes $\mc{V}$ with probability $1$.
\end{defn}

In this context, Kolmogorov's zero-one law states a two-sided phase transition point in network connectivity, implying that there is a sharp, binary phase transition in terms of network connectivity as the percolation probability changes \cite{wierman2010percolation}. This means that there is a critical threshold value of $p_c \in [0,1]$, where the nature of network connectivity changes considerably. For $p_l<p_c$ there is a subcritical regime, where a connected component spanning the entire network exists with probability $1$. The supercritical regime occurs for  $p_l>p_c$, implying that the network is fragmented, and there is no connected path spanning from one side of the network to the other \cite{mohseni2021percolation}, i.e., the probability of network connectivity is zero with certainty.

This is an essential concept for understanding how network properties change as bonds/links are randomly added or removed. For example, given an Erdos-Renyi (ER) graph \cite{erdos1960evolution} of $n$ nodes and linking probability $p \in [0,1]$ its bond-percolation threshold is equal to
\begin{align} \label{eq_pc}
p_c =1-\frac{1}{\overline{d_G}},
\end{align}
with $\overline{d_G}=\frac{(n-1)p}{2}$ as the average node degree~\cite{kawamoto2015precise,SensNets:Bollobas98}.

Properties of ER random graphs hold asymptotically. This means that an ER random graph can become \textit{asymptotically almost surely connected} for a given threshold but when examined for a small graph with the same probability, it is disconnected.

It is important to note that the connectivity properties of ER random graphs are established in the asymptotic limit. Therefore, throughout the remainder of this paper, when we refer to connectivity for ER random graphs, we specifically mean connectivity in the  \textit{asymptotically almost surely (a.a.s.)} sense. This distinction is crucial as it captures the probabilistic nature of connectivity in large networks. Furthermore, these random graph models provide powerful analytical tools for estimating percolation thresholds in practical applications, enabling us to predict and enhance the resilience of large-scale, complex networked control systems against random failures or attacks.

\section{The Proposed DRA Algorithm} \label{sec_alg}

This section presents our distributed resource allocation (DRA) algorithm designed to address the challenges of link failures and time delays in multi-agent networks. In Section~\ref{subsec:free}, we introduce the fundamental dynamics for delay-free networks and analyze its key properties, including all-time feasibility and equilibrium conditions. In Section~\ref{subsec:delay}, we extend the algorithm to handle heterogeneous time delays while preserving these essential properties. Finally, in Section~\ref{subsec:complete}, we summarize the complete algorithm and highlight its advantages compared to existing approaches in the literature.

\subsection{Solution to the Delay-Free Distributed Dynamics}\label{subsec:free}

Consider a group of $n$ networked agents which update their state $x_i(k) \in \mathbb{R}$ (state of agent $i$ at time step $k$) satisfying the following dynamics:
\begin{align}
x_i(k+1) = x_i(k) -\eta \sum_{j \in \mc{N}_i} W_{ij} g_n\Big(g_l(\partial f_i (k)) - g_l(\partial f_j (k))\Big),
\label{eq_sol}
\end{align}
with $\partial f_i=\frac{df_i}{dx_i}$, and $\eta>0$ as step-rate. The nonlinear mappings $g_n:\mathbb{R} \mapsto \mathbb{R}$ and $g_l:\mathbb{R} \mapsto \mathbb{R}$ denote possible nonlinearity at the nodes and the links, respectively.  These mappings satisfy the following properties.

\begin{ass} \label{ass_nonlin}
	The nonlinear functions $g_n$ and $g_l$ are odd and sign-preserving, and also satisfy the sector bound condition in the form $\kappa_n z < g_n(z) < \mc{K}_n z$ and $\kappa_l z< g_l(z) < \mc{K}_l z$.
\end{ass}
Examples of such nonlinear mappings in real-world applications are, just to mention a few,  the log-scale quantization given by 
	\begin{align} \nonumber
	g(z) &= \mbox{sgn}(z) \exp\Big(\rho \left[ \frac{\log(|z|)}{\rho}\right]\Big),
\end{align}	
as well as the bounded-domain saturation (clipping) defined as
\begin{align} \nonumber
	g(z) = 
	\begin{cases}
		R \mbox{sgn}(z) & |z| > R\\
		z & |z| \leq R
	\end{cases}
\end{align}
with saturation level $R$.
For these nonlinear functions the upper/lower sector bound parameters are readily obtainable; for example, for log-scale quantization we have $\kappa = 1-\frac{\rho}{2},~\mc{K} = 1+\frac{\rho}{2}$ \cite{TASE25,JFI25} and for bounded-domain saturation, given the domain $|z| \leq D_{max}$ and $D_{max} >R$, we have $\kappa = \frac{R}{D_{max}}$ and $\mc{K} = 1$. Some works in the literature even apply signum-based functions to improve the convergence rate \cite{li2023faster}.

Notice that the dynamics in~\eqref{eq_sol} satsifies the \textit{Laplacian-tracking} property, i.e., its summation satisfies the following:
\begin{equation} \label{eq_sol_L}
\kappa_n \kappa_l L \nabla F(k) \leq \sum_{j \in \mc{N}_i} W_{ij} g_n\Big(g_l(\partial f_i (k)) - g_l(\partial f_j (k))\Big) \leq \mc{K}_n \mc{K}_l L \nabla F(k).
\end{equation}

Additionally, we have the following assumption regarding the agents' exchange of information.

\begin{ass} \label{ass_net}
The agents communicate over an undirected network $\mathcal{G}(k) = {\mathcal{V}, \mathcal{E}(k)}$ with a symmetric weight matrix $W(k) = [W_{ij}(k)]_{i,j \in \mathcal{V}}$ which is (possibly) time-varying and \mbox{uniformly-connected}.
\end{ass}

Assumption~\ref{ass_net} serves two critical purposes: First, it guarantees that the proposed dynamics maintains all-time feasibility (also called anytime feasibility), thereby satisfying the equality constraint at every iteration as formally demonstrated in Lemma~\ref{lem_feas}. Second, the uniform connectivity condition is essential for ensuring algorithmic resilience against communication imperfections such as time-delays and packet drops, which will be rigorously analyzed in subsequent sections of this paper.

Before proceeding with the analysis of the algorithm's feasibility, we introduce the set of vectors that satisfy our equality constraint, which will be essential for characterizing the feasible region of our distributed resource allocation problem.

\begin{defn}\label{def_feas}
Define $\mc{S}_b$ as the set of all feasible points satisfying the equality constraint $\{\mb{x} \in \mathbb{R}^n|\mb{1}_n^\top\mb{x}  = b\}$ with $\mb{1}_n$ as all-one vector of size $n$.
\end{defn}

Having established the feasible set $S_b$ in Definition \ref{def_feas}, we now demonstrate that our proposed algorithm ensures the solution trajectory remains within this set throughout its execution\footnote{For initialization, we need $\sum_{i=1}^n x_i(0) = b$. There exist distributed algorithms to satisfy this initial constraint and locally allocate the initial values, see \cite[Algorithm~2]{cherukuri2015distributed} for example.}. The following lemma proves that the equality constraint $\sum_{i=1}^n x_i = b$ (or $\mathbf{1}_n^\top\mathbf{x} = b$) in \eqref{eq_dra} is satisfied at all iterations of the algorithm, a property we refer to as \emph{all-time feasibility}.

\begin{lem} \label{lem_feas}
	Given $g_n(\cdot)$ and $g_l(\cdot)$ satisfying Assumption~\ref{ass_nonlin} and $\mb{x}(0) \in \mc{S}_b$, $\mc{G} = \{\mc{V},\mc{E}\}$ satisfying Assumption~\ref{ass_net}, the solution under dynamics \eqref{eq_sol} is all-time feasible and the equality constraints is satisfied at all times, i.e., $\mb{x}(k) \in \mc{S}_b$ for $k>0$.
\end{lem}
\begin{proof}
	At any time $k>0$, we have
	\begin{align}
		\mb{1}_n^\top \mb{x}(k+1) = \mb{1}_n^\top \mb{x}(k)
		-\eta \sum_{i = 1}^n \sum_{j \in \mc{N}_i} W_{ij} g_n\Big(g_l(\partial f_i (k)) - g_l(\partial f_j (k))\Big).
		\label{eq_sol_proof}
	\end{align}
	We have the functions $g_n(\cdot)$ and $g_l(\cdot)$ as sign-preserving odd mappings and we also have $W_{ij}=W_{ji}$. This implies that the dual summation in \eqref{eq_sol_proof} is equal to zero and $\mb{1}_n^\top \mb{x}(k+1) = \mb{1}_n^\top \mb{x}(k)$. Recalling that $\mb{x}(0) \in \mc{S}_b$ this implies that $\mb{x}(k) \in \mc{S}_b$ for all $k>0$.
\end{proof}

This critical all-time feasibility property guarantees that the system can be terminated at any point while preserving the feasibility of the equality constraint. In practical distributed applications such as economic dispatch and generator coordination, this property is particularly valuable as it ensures continuous service without disruption -- the power demand-supply balance is preserved throughout the algorithm's execution regardless of when termination occurs.

Furthermore, we can establish that the equilibrium $\mathbf{x}^*$ of the proposed dynamics  \eqref{eq_sol} must satisfy $\nabla F(\mathbf{x}^*) \in \text{span}(\mathbf{1}_n)$. This condition has important implications: any point where $\nabla F(\mathbf{x}^*) \notin \text{span}(\mathbf{1}_n)$ cannot serve as an equilibrium of our system \eqref{eq_sol}. The formal proof of this characteristic is presented in the following lemma.

\begin{lem} \label{lem_eqb}
Under the uniform-connectivity of undirected multi-agent network $\mc{G}$, the equilibrium $\mb{x}^*$ of the dynamics \eqref{eq_sol} is in the form $\nabla F(\mb{x}^*) \in \mbox{span}(\mb{1}_n)$.
\end{lem}
\begin{proof}
We prove this by contradiction. Suppose that  $\widehat{\mb{x}}$ is the equilibrium of \eqref{eq_sol}, i.e., $\widehat{x}_i(k+1) = \widehat{x}_i(k)$ for $\forall i$, for which $\nabla F(\widehat{\mb{x}}) \notin \mbox{span}(\mb{1}_n)$. This implies that there exists (at least) two agents $a$ and $b$ such that $\partial f_a(\widehat{x}_a) \neq \partial f_b(\widehat{x}_b)$. Following the definition of connectivity, there exists a path over the connected (union) network $\mc{G}_T$ between these two agents  $a,b$. Then, for two neighboring agents $j \in \mc{N}_i$ on this path we have $\partial f_i(\widehat{x}_i) \neq \partial f_j(\widehat{x}_j)$. Therefore, invoking~\eqref{eq_sol}, we have
\begin{align}
	\sum_{j \in \mc{N}_i} W_{ij} g_n\Big(g_l(\partial f_i (k)) - g_l(\partial f_j (k))\Big) \neq 0,
\end{align}
and, therefore, $\widehat{x}_i(k+1) \neq \widehat{x}_i(k)$. This violates the definition of equilibrium for $\widehat{\mb{x}}$, and the conclusion follows.
\end{proof}

\subsection{Solution in the Presence of Network Delays}\label{subsec:delay}

In the previous subsection, we presented a delay-free distributed algorithm that possesses the desirable properties of all-time feasibility and optimal convergence under uniform network connectivity. However, real-world communication networks often experience unavoidable time delays due to various factors such as packet retransmission, processing time, congestion, and physical constraints. These delays can significantly impact algorithm performance and potentially compromise both the feasibility and optimality of traditional distributed solutions.

We now extend our approach to address heterogeneous time delays while preserving the key properties of our delay-free algorithm. This extension is crucial for practical implementations, as it ensures the resilience of our distributed resource allocation method in realistic networking conditions. The following assumption formalizes the delay model we consider in the sequel.

\begin{ass}
The time delays are, in general, heterogeneous at different links and at different time instants $k \in \mathbb{Z}^+$. The time delay at a bidirectional link $(i,j)$ where $i,j \in \{1,2,\ldots,n\}$ is denoted by $\tau_{ij}(k) \in \{0,1,2,\ldots,\overline{\tau}\} \leq \overline{\tau}$ and, in the most general form, is arbitrary, random, and bounded by a maximum possible delay $\overline{\tau} \in \mathbb{Z}^+$.
\end{ass}

Additionally, we assume the information sent from the sender node $i$ to recipient node $j$ is \textit{time-stamped} and, thus, the delay is known to the node~$j$.

Therefore, in the presence of network delays, the proposed agents'  update should satisfy the following dynamics:
\begin{align} \nonumber
x_i(k+1) = x_i(k) -{\eta_\tau} \sum_{j \in \mathcal{N}_i} &\sum_{r=0}^{\overline{\tau}} W_{ij}g_n\Big(g_l(\partial f_i (k-r))\\
&\qquad - g_l(\partial f_j (k-r))\Big)\mathcal{I}_{k-r,ij}(r),
\label{eq_sol_delay}
\end{align}
with $\eta_\tau$ as the step rate in the presence of time-delays.
Besides, to capture the presence and timing of delays in our proposed solution, we employ an indicator function $\mathcal{I}$ defined as
\begin{align} \label{eq_mcI}
\mathcal{I}_{k,ij}(\tau) = \left\{ \begin{array}{ll}
1, & \text{if}~ \tau_{ij}(k) = \tau,\\
0, & \text{otherwise}.
\end{array}\right.
\end{align}

In the worst case of unknown asymmetric time delays, both agents $i,j$ can update their mutual information based on the maximum delay at every link or the maximum possible delay $\overline{\tau}$, i.e., $\mathcal{I}_{k,ij}(\overline{\tau}) = \mathcal{I}_{k,ji}(\overline{\tau}) = 1$.  Later, we will also make explicit that $\eta_\tau$ depends on the time-delay bound $\overline{\tau}$.

Using the same arguments as in Lemma~\ref{lem_feas} and~\ref{lem_eqb}, we can establish that both the all-time feasibility and equilibrium properties extend naturally to the time-delayed dynamics~\eqref{eq_sol_delay}. This connection arises from a fundamental insight: \emph{delayed information exchange effectively manifests as temporal changes in the network topology}.

Consider a concrete example: when information exchange between nodes~$i$ and $j$ experiences a delay from time $k$ to time $k+\tau$, we can conceptualize this as removing link $(i,j)$ from $\mathcal{G}(k)$ and adding it to $\mathcal{G}(k+\tau)$ instead. This topological reframing allows us to leverage our existing analysis framework.

By recasting time delays as topology modifications, we can directly extend the properties established for the original dynamics~\eqref{eq_sol} to the delay-inclusive dynamics~\eqref{eq_sol_delay}. The mathematical structure remains consistent, with the primary difference being the temporal distribution of connections.

Furthermore, since all delays are bounded by $\overline{\tau}$, we can guarantee that the effective topology associated with the delayed dynamics maintains uniform connectivity over a slightly extended time window. Specifically, the union network $\mathcal{G}_{T+\overline{\tau}}(k) = \cup_{k}^{k+T+\overline{\tau}} \mathcal{G}(k)$ remains connected, preserving the essential connectivity property required for convergence, albeit over a longer time horizon that accommodates the maximum possible delay.

This topological perspective on delays provides a unifying framework for analyzing the resilience of our algorithm without requiring fundamentally different mathematical machinery for the delayed case.

\subsection{The Complete Resilient DRA Algorithm}\label{subsec:complete}

Our proposed distributed allocation and scheduling algorithm is presented in Algorithm~\ref{alg_ac}. The method follows a structured approach with two key phases per iteration:

First, each node exchanges gradient information derived from its local objective function with neighboring nodes across the network $\mathcal{G}$. This information sharing forms the foundation of our distributed approach, enabling coordination without centralized control.

Second, nodes update their states according to the resilient DRA dynamics in~\eqref{eq_sol_delay}, which explicitly accounts for heterogeneous time delays across communication links. This update mechanism iteratively guides the system toward the optimal solution of the constrained optimization problem while maintaining all-time feasibility.

\begin{rem}
For the special case where no communication delays are present (i.e., $\tau_{ij}(k) = 0$ for all links), our algorithm naturally simplifies to the delay-free dynamics in~\eqref{eq_sol}, demonstrating the unified nature of our theoretical framework. This adaptive behaviour allows the algorithm to seamlessly operate in both ideal and delay-prone network environments.
\end{rem}

%\begin{algorithm}[t]
%	\caption{The DRA Algorithm}
%	\begin{algorithmic}[1]
%		\State \textbf{Input:}
%		\State \textbf{Initialization:}
%		\State \textbf{While} {termination criteria NOT hold} \textbf{do}
%		\State Each node $i$ receives gradient data from $j\in \mc{N}_i$
%		\State Each node $i$	updates $x_i(k+1)$ via Eq.~\eqref{eq_sol} (or via Eq.~\eqref{eq_sol_delay} in the presence of time-delays)
%		\State Each node $i$ shares its gradient information with its neighbors $j \in \mc{N}_i$
%		\State $k \leftarrow k+1$
%		\State \textbf{End While}
%		\State \textbf{Output:}   Assigned resources $\mb{x}(k)$ and associated cost $F(\mb{x}(k))$	
%	\end{algorithmic}
%	\label{alg_ac}
%\end{algorithm}

\begin{algorithm} \label{alg_ac}
\textbf{Given:}  Uniformly connected network $\mc{G}$,  $W$,  $\eta$, cost $f_i(\cdot)$  \\	
\textbf{Initialization:} $k=0$, $\mb{x}(0) \in \mc{S}_b$
\\
\While{termination criteria NOT hold}{
	Each node $i$ receives gradient data from $j\in \mc{N}_i$ \;
	Each node $i$	updates $x_i(k+1)$  via Eq.~\eqref{eq_sol_delay}\;
	Each node $i$ shares its gradient information with its neighbors $j \in \mc{N}_i$ \;
	$k \leftarrow k+1$ \;
}
\textbf{Return:}  Assigned resources $\mb{x}(k)$ and associated cost $F(\mb{x}(k))$\;	
\caption{The Resilient DRA Algorithm. }
\end{algorithm}

\begin{rem}
	In the absence of node nonlinearity (i.e., linear mapping $g_n(u)=u$), the algorithm converges over general weight-balanced directed networks.
\end{rem}

To complete our algorithmic analysis, we now establish the computational requirements of our approach through a formal complexity characterization.

\begin{prop}[Computational Complexity] \label{prop:complexity}
The computational complexity of the proposed DRA algorithm is $\mathcal{O}(kn^2)$ for $k$ iterations in a network of $n$ agents when $g_n$ and $g_l$ are linear functions. When $g_n$ and $g_l$ are nonlinear, the complexity becomes $\mathcal{O}(kn^2 + kn\cdot C(g))$, where $C(g)$ denotes the computational cost of evaluating the nonlinear functions.
\end{prop}

\begin{proof}
We analyze the computational complexity by examining the operations performed at each iteration of Algorithm~\ref{alg_ac} across all nodes:

\begin{enumerate}
    \item Local gradient computation: Each agent $i$ computes $\partial f_i(x_i(k))$, requiring $\mathcal{O}(1)$ operations per agent, for a total of $\mathcal{O}(n)$ operations across the network.

    \item Information exchange: In the worst case, the network is fully connected, and each agent communicates with all other agents. This requires $\mathcal{O}(n)$ message transmissions per agent, for a total of $\mathcal{O}(n^2)$ communication operations.

    \item State update: For each agent $i$, the state update in Equation~\eqref{eq_sol} or~\eqref{eq_sol_delay} requires:
    \begin{itemize}
        \item When $g_n$ and $g_l$ are linear: Computing the weighted sum of gradient differences across all neighbors requires $\mathcal{O}(|\mathcal{N}_i|)$ operations, where $|\mathcal{N}_i|$ is the number of neighbors of agent $i$. In the worst case, $|\mathcal{N}_i| = n-1$, leading to $\mathcal{O}(n)$ operations per agent and $\mathcal{O}(n^2)$ operations across all agents.

        \item When $g_n$ and $g_l$ are nonlinear: Each evaluation of $g_n$ or $g_l$ adds a computational cost of $C(g_n)$ or $C(g_l)$, respectively. With $\mathcal{O}(n)$ such evaluations per agent, this adds $\mathcal{O}(n \cdot C(g))$ operations per agent and $\mathcal{O}(n^2 \cdot C(g))$ across all agents, where $C(g) = \max\{C(g_n), C(g_l)\}$.
    \end{itemize}
\end{enumerate}

For a single iteration, the complexity is dominated by the state update step, resulting in $\mathcal{O}(n^2)$ when $g_n$ and $g_l$ are linear, and $\mathcal{O}(n^2 + n^2 \cdot C(g)) = \mathcal{O}(n^2 \cdot (1+C(g)))$ when they are nonlinear.

Over $k$ iterations, the total computational complexity becomes $\mathcal{O}(kn^2)$ for the linear case and $\mathcal{O}(kn^2 \cdot (1+C(g)))$ for the nonlinear case, which can be simplified to $\mathcal{O}(kn^2 + kn^2 \cdot C(g))$ or $\mathcal{O}(kn^2 + kn \cdot C(g))$ depending on whether we consider the per-agent or network-wide complexity for the nonlinear function evaluations.
\end{proof}

\begin{rem}[Resilience to Network Disruptions] \label{rem:resilience}
Algorithm~\ref{alg_ac} significantly advances the state-of-the-art in distributed resource allocation through two critical innovations that enhance resilience to network disruptions:

\begin{enumerate}
    \item \textbf{Relaxed Weight Matrix Requirements:} Our proposed dynamics in~\eqref{eq_sol} require only that the weight matrix $W$ be symmetric, eliminating the more restrictive weight-stochastic condition imposed by previous approaches~\cite{boyd2006optimal,falsone2020tracking,rikos2021optimal,Carli_ADMM,wang2020dual}. This relaxation provides substantial benefits during link failures:

    \begin{itemize}
        \item No weight compensation strategies are needed when links fail or network topology changes;
        \item Unlike stochastic weight matrices that necessitate complex redistribution mechanisms~\cite{6426252,cons_drop_siam} after topology changes; and
        \item Weight-balanced conditions are naturally preserved under link removals, maintaining algorithm stability.
    \end{itemize}

    Figure~\ref{fig_remov} illustrates this advantage: when the red link fails, the network loses stochasticity (observe weights at nodes $v_1$ and $v_2$), but maintains weight-balance, allowing our algorithm to proceed without reconfiguration.

    \begin{figure}[h]
        \centering
        \includegraphics[width=1.6in]{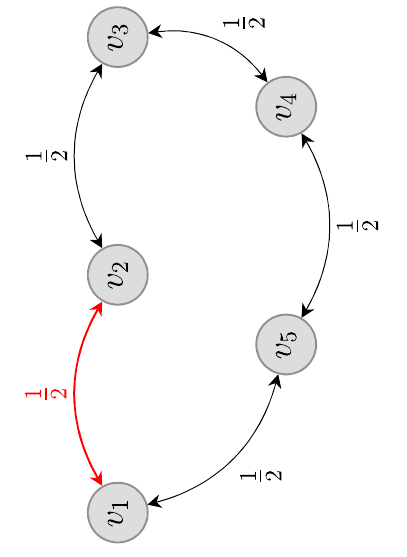}
        \includegraphics[width=1.6in]{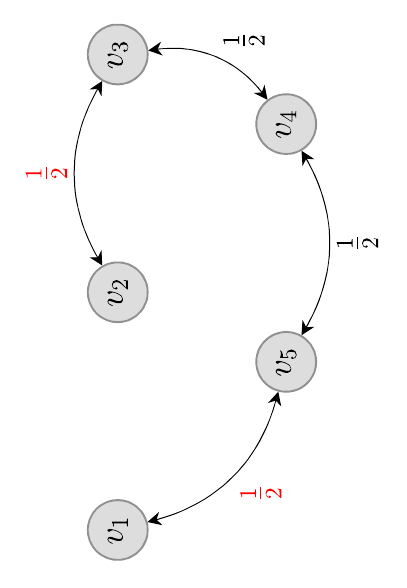}
   \caption{Comparison of network properties under link failure. Left: Original network with both weight-balance and weight-stochasticity properties. Right: After failure of the red link, weight-stochasticity is compromised at nodes $v_1$ and $v_2$, while weight-balance is maintained, illustrating our algorithm's inherent resilience to topology changes.}
        \label{fig_remov}
    \end{figure}

    \item \textbf{Uniform Connectivity Guarantee:} Our algorithm converges under the significantly weaker condition of uniform connectivity rather than requiring all-time connectivity as in~\cite{boyd2006optimal,cherukuri2016initialization,rikos2021optimal,cdc_dtac,banjac2019decentralized,jiang2022distributed,wang2020dual,Carli_ADMM,falsone2020tracking}. This advancement enables the following:

    \begin{itemize}
        \item Tolerance of temporary network disconnections caused by high packet drop rates;
        \item Successful operation in sparse or intermittently connected network topologies; and
        \item Guaranteed convergence even when connectivity is only established periodically every $T \in \mathbb{Z}_{\geq 0}$ timesteps.
    \end{itemize}
\end{enumerate}
\end{rem}

In Section~\ref{sec_conv},  we rigorously prove that the proposed relaxed connectivity requirement is sufficient to ensure algorithm convergence while dramatically improving resilience to practical network imperfections. Thus, maintaining optimal performance in realistic networking environments where disruptions are common, representing a significant practical advancement over existing approaches.

\section{Convergence Analysis} \label{sec_conv}
In this section, we present a comprehensive convergence analysis of our proposed resilient distributed resource allocation algorithm (Algorithm~\ref{alg_ac}). First, in Section 5.1, we establish fundamental convergence guarantees over uniformly-connected networks, deriving explicit bounds on the step-rate parameter to ensure optimal convergence. Next, in Section 5.2, we extend our analysis to networks experiencing random link failures, leveraging percolation theory to quantify resilience thresholds and convergence properties under various failure rates. Finally, in Section 5.3, we address the critical challenge of heterogeneous time delays across communication links, providing theoretical guarantees for convergence despite communication latency and developing practical bounds on maximum tolerable delays.

\subsection{General Proof of Convergence} \label{sec_conv1}
In the following theorem, we prove that Algorithm~\ref{alg_ac} converges to the optimal solution when the step-rate $\eta$ is sufficiently small, with the specific bound on $\eta$ dependent on the spectral properties $\lambda_{2T}$ and $\lambda_{nT}$ of the Laplacian matrix associated with the uniformly-connected network $\mathcal{G}_T$.

\begin{theorem} \label{thm_conv}
Under Assumptions~\ref{ass_nonlin} (sector-bound nonlinearity) and Assumption~\ref{ass_net} (uniformly-connected network), Algorithm~\ref{alg_ac} converges to the optimal solution of problem~\eqref{eq_dra} when the step-rate satisfies
\begin{equation}
\eta < \overline{\eta} := \frac{\kappa_n \kappa_l \lambda_{2T}}{u \lambda_{nT}^2 K_n^2 K_l^2 (T+1)},
\end{equation}
where $\lambda_{2T}$ and $\lambda_{nT}$ are the smallest non-zero and largest eigenvalues of the Laplacian matrix associated with $\mathcal{G}_T$, respectively.
\end{theorem}
\begin{proof}
The all-time feasibility of the solution has been established in Lemma~\ref{lem_feas}. To prove convergence, we define the residual function $\overline{F}(k) := F(\mathbf{x}(k))-F(\mathbf{x}^*)$ and the state difference $\delta \mathbf{x}(k) := \mathbf{x}(k+T)-\mathbf{x}(k)$.

Using the smoothness property $0 < \frac{d^2f_i}{dx_i^2} < 2u$, we can bound the change in the objective function as
\begin{align}
	F(\mathbf{x}(k+T)) \leq F(\mathbf{x}(k))
	+ \nabla F(\mathbf{x}(k))^\top \delta \mathbf{x}(k) +  u\delta \mathbf{x}(k)^\top \delta \mathbf{x}(k).
	\label{eq_taylor_2}
\end{align}

To establish that the residual is non-increasing, i.e., $\overline{F}(k+T) \leq \overline{F}(k)$, it suffices to show that
\begin{align} \label{eq_proof1}
	\nabla F(\mathbf{x}(k))^\top \delta \mathbf{x}(k) +  u\delta \mathbf{x}(k)^\top \delta \mathbf{x}(k)  \leq 0.
\end{align}

For notational simplicity, we omit the time index $k$ in the remainder of the proof. Define the dispersion parameter $\xi := \nabla F - \frac{\mathbf{1}_n^\top \nabla F}{n} \mathbf{1}_n \in \mathbb{R}^n$, which measures the deviation of the gradient from consensus.

From the proposed dynamics \eqref{eq_sol}, the Laplacian-tracking property \eqref{eq_sol_L}, and the sector-bound conditions on $g_n(\cdot)$ and $g_l(\cdot)$, we obtain
\begin{align} \label{eq_proof_dx}
	-\eta(T+1)\kappa_n\kappa_l L\nabla F \geq \delta \mathbf{x} \geq -\eta(T+1)K_n K_l L\nabla F.
\end{align}

This allows us to bound the quadratic term in \eqref{eq_proof1} as follows:
\begin{align}
	\delta \mathbf{x}^\top \delta \mathbf{x} &\leq \eta^2 (T+1)^2 K_n^2 K_l^2 \nabla F^\top L^\top L \nabla F \nonumber \\
	&\leq \eta^2 (T+1)^2 K_n^2 K_l^2 \lambda_{nT}^2 \xi^\top\xi,
	\label{eq_proof_dx2}
\end{align}
where the second inequality follows from Lemma~\ref{lem_xLy} and Corollary~\ref{cor_xLy_union}.

Similarly, for the first term in \eqref{eq_proof1}, Lemma~\ref{lem_xLy2} and Corollary~\ref{cor_xLy_union} yield
\begin{align} \label{eq_proof_df}
	\nabla F^\top \delta \mathbf{x} \leq -\kappa_n \kappa_l \eta (T+1) \lambda_{2T} \xi^\top\xi.
\end{align}

Combining \eqref{eq_proof_dx2} and \eqref{eq_proof_df}, the condition in \eqref{eq_proof1} is satisfied when
\begin{align}
	\label{eq_proof_rho}
	(-\kappa_n \kappa_l \eta \lambda_{2T} (T+1) + u \lambda_{nT}^2 K_n^2 K_l^2 \eta^2(T+1)^2)\xi^\top\xi \leq 0.
\end{align}

This inequality holds strictly when
\begin{align} \label{eq_eta}
	\eta < \frac{\kappa_n \kappa_l \lambda_{2T}}{u \lambda_{nT}^2 K_n^2 K_l^2 (T+1)}.
\end{align}

Thus, under the stated condition on the step-rate $\eta$, the residual $\overline{F}(k)$ serves as a discrete Lyapunov function that decreases monotonically. The equality in \eqref{eq_eta} is attained only when $\xi = \mathbf{0}_n$, which corresponds to the equilibrium condition $\nabla F \in \text{span}(\mathbf{1}_n)$; thus, completing the proof.
\end{proof}

Note that all the parameters involved in Eq.~\eqref{eq_eta}, including the ones related to the nonlinear mappings, are readily obtainable and thus one can easily select the step size $\eta$ for convergence.

\begin{rem} \label{rem_spectrum}
    The literature on distributed systems and consensus algorithms provides various characterizations of Laplacian spectral properties based on graph topology~\cite{SensNets:Olfati04,graph_handbook}. For instance, \cite{SensNets:Olfati04} relates the spectral range to node degrees, while for a network with diameter $e_{\mathcal{G}}$, a lower bound on algebraic connectivity can be established as $\lambda_2(\mathcal{G}) \geq \frac{1}{ne_{\mathcal{G}}}$~\cite[p. 571]{graph_handbook}. These spectral bounds are particularly valuable for predicting algorithm performance in practical network deployments.
\end{rem}

    Notice that the convergence rate of Algorithm~\ref{alg_ac} is governed by three key factors: (\emph{i}) the step-rate parameter $\eta$, which must be properly calibrated according to Theorem~\ref{thm_conv}; (\emph{ii}) the network connectivity quantified by the algebraic connectivity $\lambda_2$, with denser networks exhibiting faster convergence; and (\emph{iii}) the sector-bound parameters of the nonlinear functions $g_n$ and $g_l$, which affect information flow through the network. These relationships provide practical guidance for tuning algorithm parameters to achieve desired performance in specific application scenarios. 
        %For linear case (i.e., $g_n(z)=g_l(z)=z$) and \textit{strongly-convex} cost functions satisfying $2v<\frac{d^2f_i}{dx_i^2} < 2u$, the convergence is approximated as geometric rate of $\exp(-4v \lambda_{2T})$. However, 
    For the general case, the convergence rate of the proposed dynamics depends on the properties of the nonlinear mappings $g_n,g_l$ and heterogeneous time-delays. Therefore, the exact convergence rate cannot be easily obtained.

\subsection{Convergence under Link Failure} \label{sec_conv2}In this section, we consider the general case of random link failure in network topology, which frequently occurs due to packet drops in practical communication systems.

Similar to the modelling of time-delays discussed in Section~\ref{subsec:delay}, we model packet drops in network $\mathcal{G}(k)$ at time $k$ as changes in the network topology; specifically, as temporary link removals from the topology. These links may disappear and reconnect dynamically, with different links potentially failing at different time instants.

The key insight in our approach is that uniform network connectivity can be maintained despite link failures if the remaining network preserves connectivity over $T$ iterations, ensuring that $\mathcal{G}_T(k)$ remains connected. This assumption on uniform-connectivity of $\mathcal{G}_T(k)$ is central to our analysis of resilience against link failures.

Recall that the union network $\mathcal{G}_T(k) = \cup_{k}^{k+T} \mathcal{G}(k)$ represents the combined topology of all links present at any point over the time interval $T$. For a given~$T$, a link between nodes $i$ and $j$ is removed from $\mathcal{G}_T(k)$ only if it is absent in all $T+1$ network instances $\mathcal{G}(k),\mathcal{G}(k+1),\ldots,\mathcal{G}(k+T)$. In the extreme case, one can conceptualize the union network $\mathcal{G}_T(k)$ as having $T+1$ potential instances of each link between nodes $i$ and $j$. A link $(i,j)$ is removed from $\mathcal{G}_T(k)$ only if all these $T+1$ instances fail.

Now, let $p_l$ denote the probability of link failure in $\mathcal{G}(k)$ at any given time. Intuitively, the probability that link $(i,j)$ fails in the union network $\mathcal{G}_T(k)$ is~$p_l^{T+1}$, representing the likelihood of failure across all $T+1$ time instances.

Following the definition of bond-percolation threshold $p_c$ from Section~\ref{sec_perc}, when $p_l < p_c$, the network $\mathcal{G}(k)$ remains connected with probability 1 despite random link failures. However, the challenging case occurs when $p_l > p_c$, where instantaneous connectivity is lost. This scenario is addressed in the following theorem.

\begin{theorem}[Resilience to Link Failures] \label{thm_drop}
Let $\mathcal{G}(k)$ be a connected network with bond-percolation threshold $p_c$. Under random link failures occurring with probability $p_l$, where $p_c < p_l < 1$, the union network $\mathcal{G}_{T}(k)$ maintains uniform connectivity with probability 1 for any time window $T$ satisfying
\begin{align} \label{eq_B*}
p_l^{T+1} < p_c.
\end{align}
\end{theorem}

\begin{proof}
Given a network with bond-percolation threshold $p_c$, we first establish the existence of a minimum time window $T^*$ that ensures continued connectivity despite link failures occurring with probability $p_l > p_c$.

Let $T^* \in \mathbb{Z}_{> 0}$ denote the minimum value satisfying $p_l^{1+T^*} < p_c$. Such a value must exist since (i) $p_l < 1$ by assumption, making $p_l^{1+T}$ strictly less than 1 for all $T > 0$, (ii) $\lim_{T \to \infty} p_l^{1+T} = 0$ for $p_l < 1$, and (iii)
$p_c > 0$ by the definition of percolation threshold.

The function $p_l^{1+T}$ has two important properties: (a)
 it is monotonically decreasing with respect to $T$ (since $p_l < 1$), and (b) it is monotonically increasing with respect to $p_l$ (for any fixed $T > 0$).

Therefore, for any $T \geq T^*$, we have $p_l^{1+T} \leq p_l^{1+T^*} < p_c$. This implies that the probability of a link failing in all $T+1$ consecutive time steps is less than the percolation threshold $p_c$.

By the definition of bond percolation (Section~\ref{sec_perc}), when the link failure probability is less than $p_c$, the network remains connected with probability~$1$. Consequently, for any $T \geq T^*$ and $p_c < p_l < 1$, the union network $\mathcal{G}_{T}(k)$ maintains connectivity with probability 1, thereby ensuring uniform connectivity over the time window $T$.
\end{proof}

\begin{cor}[Convergence Under Link Failure] \label{cor_link_failure}
For a network experiencing random link failures with probability $p_l$ where $p_c < p_l < 1$, Algorithm~\ref{alg_ac} converges to the optimal solution with probability 1 when
\begin{enumerate}
    \item The time window $T$ satisfies $p_l^{T+1} < p_c$ as per Theorem~\ref{thm_drop}, and
    \item The step-rate $\eta$ satisfies the bound in Equation~\eqref{eq_eta} from Theorem~\ref{thm_conv}.
\end{enumerate}
\end{cor}

The resilience of our algorithm to link failures comes with a performance trade-off. As the link failure probability $p_l$ increases, the network becomes more sparse, reducing the algebraic connectivity $\lambda_2$. This spectral degradation has direct implications on algorithm performance; specifically, \eqref{eq_eta} indicates that smaller values of $\lambda_2$ necessitate smaller step-rates $\eta$ to maintain convergence guarantees, resulting in slower overall convergence.

Hence, the relationship between link failure rate and convergence speed establishes a quantifiable resilience-performance trade-off that can guide parameter selection in practical deployments. The algorithm can tolerate substantial link failures while maintaining convergence guarantees, albeit at reduced convergence rates that are precisely characterized by the spectral properties of the resulting network topology.

\subsection{Convergence under Time-Delays} \label{sec_conv3}

Latency is an unavoidable challenge in practical communication networks, arising from packet retransmission, processing bottlenecks, and physical transmission constraints. This section analyzes the convergence properties of our proposed algorithm under heterogeneous time delays across the multi-agent network.

For the dynamics defined in~\eqref{eq_sol_delay}, we can extend our theoretical analysis by examining the union network $\mc{G}_{T+\overline{\tau}}(k) = \cup_{j=k}^{k+T+\overline{\tau}} \mc{G}(j)$. This formulation reveals a crucial insight: when the delay-free network $\mc{G}_{T}(k)$ exhibits connectivity over $T$ time steps, we can establish connectivity of $\mc{G}_{T+\overline{\tau}}(k)$ even when communication is subject to heterogeneous time delays bounded by $\overline{\tau}$. This perspective allows us to address arbitrary, time-varying, and heterogeneous delays within a unified theoretical framework.

\begin{theorem} \label{thm:delay_convergence}
Given a uniformly-connected network $\mc{G}_{T}(k)$, Algorithm~\ref{alg_ac} converges to the optimal solution of problem~\eqref{eq_dra} for a sufficiently small step-rate $\eta_{\tau}$ satisfying
\begin{equation} \label{eq:delay_bound}
\eta_{\tau} < \frac{\kappa_n\kappa_l\lambda_{2T}}{u\lambda^2_{nT}K^2_nK^2_l(T + \overline{\tau} + 1)},
\end{equation}
where $\overline{\tau}$ represents the upper bound on heterogeneous time-delays across all communication links.
\end{theorem}

\begin{proof}
The proof extends the analysis of Theorem~\ref{thm_conv} to the uniformly-connected network $\mc{G}_{T+\overline{\tau}}(k)$. A key observation is that the eigenspectrum of $\mc{G}_{T+\overline{\tau}}(k)$ is identical to that of $\mc{G}_{T}(k)$, since the same set of communication links is distributed over an extended time scale $T + \overline{\tau}$.

Specifically, for a network subject to time delays bounded by $\overline{\tau}$, the set of links $V_{T+\overline{\tau}} = \cup_{j=k}^{k+T+\overline{\tau}} V(j)$ corresponds exactly to the set of links $V_T$ in the delay-free case over time period $T$. This topological equivalence allows us to apply the same Lyapunov analysis as in Theorem~\ref{thm_conv}, with the adjustment that we consider the extended time horizon $T + \overline{\tau}$ instead of $T$.

Following the same sequence of algebraic manipulations with the modified time horizon, we arrive at the convergence condition in~\eqref{eq:delay_bound}, which guarantees that the residual function serves as a discrete Lyapunov function that decreases monotonically.
\end{proof}

The bound in Theorem~\ref{thm:delay_convergence} characterizes the fundamental trade-off between delay tolerance and convergence rate. For a given step-rate $\eta_{\tau}$, we can determine the maximum tolerable delay $\overline{\tau}$ for guaranteed convergence given by
\begin{equation} \label{eq_tau_eta}
\overline{\tau} < \frac{\kappa_n\kappa_l\lambda_{2T}}{u\eta_{\tau}\lambda^2_{nT}K^2_nK^2_l} - 1 - T.
\end{equation}

This expression provides the upper bound on network delays for which Algorithm~\ref{alg_ac} maintains convergence guarantees despite heterogeneous, time-varying delays. In scenarios with unknown but bounded delays, agents can synchronize their information updates at intervals of $\overline{\tau}$ time steps. Under this implementation, all communication links in $\mc{G}(k)$ effectively transfer information to time $k + \overline{\tau}$, preserving uniform connectivity over the original time period $T$. This approach enables a tighter upper bound on the step-rate, analogous to the bound in~\eqref{eq_eta}.

\begin{rem} \label{rem_delay}
Equation~\eqref{eq_tau_eta} reveals a fundamental resilience versus performance trade-off. Specifically, accommodating larger time delays $\overline{\tau}$ necessitates reducing the step-rate $\eta_{\tau}$ to maintain convergence guarantees. This reduction in step-rate inevitably leads to slower convergence. This trade-off provides system designers with a quantitative framework for balancing delay tolerance against convergence speed in practical deployments.
\end{rem}

The analysis presented in this subsection demonstrates that our algorithm can tolerate substantial communication delays while maintaining both feasibility and optimality guarantees -- a critical advancement over existing approaches that require either continuous connectivity or synchronous, delay-free communication.

\section{Simulation} \label{sec_sim}
In this section, we provide simulations to validate the performance of Algorithm~\ref{alg_ac} in the presence of link failure and time delays. For the simulation we optimize, in a distributed way, the following cost function (taken from \cite{doan2017ccta}) at every agent $i$:
\begin{align} \label{eq_nonquad}
f_i(x_i) =  \omega_i(x_i - \alpha_i)^4 + \sigma ([x_i-M_i]^++[m_i-x_i]^+),
\end{align}
with penalty terms $ \sigma ([x_i-M_i]^++[m_i-x_i]^+)$ and $[u]^+ = (\max\{0,u\})^2$ addressing the box constraints. For the simulation, we consider $\sigma=20$ large enough to better address the box constraints with $m_i =1$, $M_i=10$. Other simulation parameters are chosen as follows: $\alpha_i \in (0,2]$, $\omega_i \in (0,0.02]$, and $b=200$ as the sum of resources.

\subsection{Convergence under Dynamic Network Topology}
\begin{figure}[]
\centering
\includegraphics[width=2.5in]{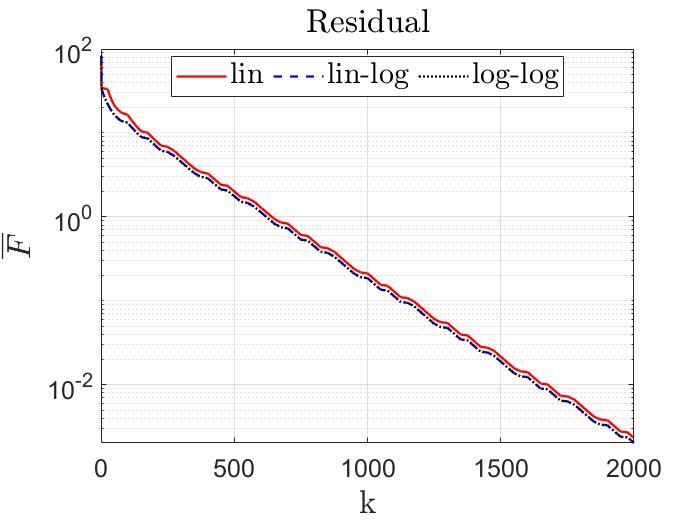}
\includegraphics[width=2.5in]{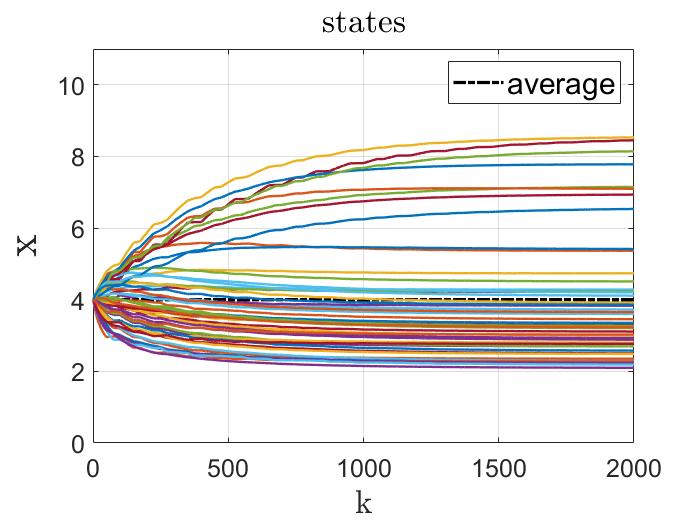}
\caption{Convergence behaviour under dynamic network topology. Left: Evolution of cost residual with penalty term $\sigma([x_i-M_i]^++[m_i-x_i]^+)$ for both linear and nonlinear (log-scale quantization) models in a non-persistently connected network. Right: State trajectories of all agents under log-scale quantization, demonstrating constant average values and confirming the all-time feasibility of the proposed solution.}
\label{fig_dyn}
\end{figure}
\begin{figure}[]
\centering
\includegraphics[width=2.5in]{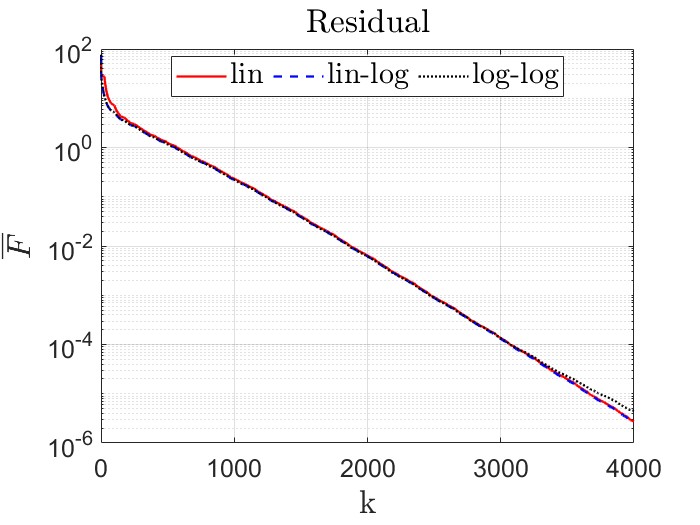}
\includegraphics[width=2.5in]{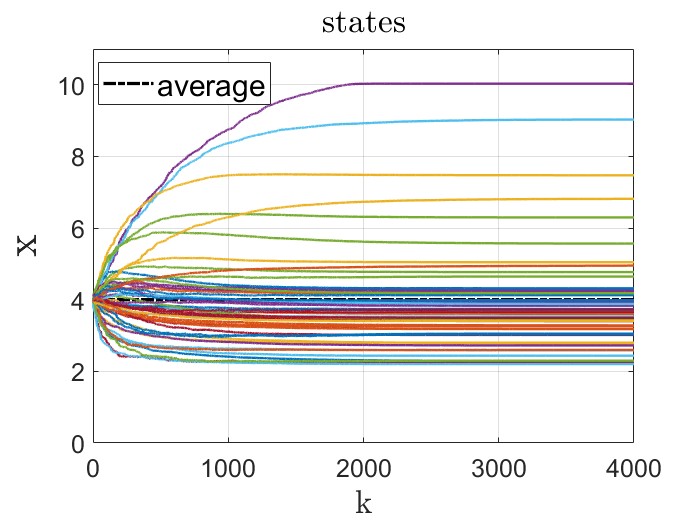}
\caption{Algorithm performance with logarithmic penalty function. Left: Convergence of cost residual with smoothed penalty term $\frac{1}{\mu}(\log(1+\exp(\mu(x_i-M_i)))+\log(1+\exp(\mu(m_i-x_i))))$ under dynamic network conditions. Right: State evolution of all agents maintaining constant average value, confirming all-time feasibility with logarithmic constraints.}
\label{fig_dyn2}
\end{figure}
The multi-agent network is considered as an ER random graph of $n=50$ nodes with four different linking probability $p\in \{20\%,10\%,5\%,1\%\}$ to model a dynamic network with random entries $W_{ij}$. Note that, for $p=1\%$, we have $p<\frac{1}{n-1}$ implying that the network is disconnected, since the network switches between these four possible topologies every $25$ iterations.

For the sector-bound nonlinearity, we consider log-scale quantization defined as $q(x)=\mbox{sgn}(x) \exp(\epsilon[\frac{\log(|x|)}{\epsilon}])$ with $\epsilon$ as the quantization level. It is known that logarithmic quantization satisfies $x(1-\frac{\epsilon}{2}) \leq q(x)\leq x(1+\frac{\epsilon}{2})$.

The distributed solution follows~\eqref{eq_sol} with three different options of nonlinearity at the node and link dynamics: (\emph{i}) \mbox{linear-linear}: $g_n(x)=x$ and $g_l(x)=x$;  (\emph{ii}) linear-logarithmic: $g_n(x)=x$ and  $g_l(x)=q(x)$; and  (\emph{iii}) logarithmic-logarithmic: $g_n(x)=q(x)$ and $g_l(x)=q(x)$.

For our implementation, we choose $\epsilon_l = \frac{1}{8}$ and $\epsilon_n = \frac{1}{1024}$ as the link and node quantization levels, respectively. We set $\eta =0.1$ for the simulation, and the results are shown in Fig.~\ref{fig_dyn}. The average of states $\frac{1}{n} \sum_{i=1}^n x_i$ remains constant equal to $\frac{b}{n}$ throughout the evolution of dynamics, demonstrating the all-time feasibility  of the proposed solution, i.e., $\sum_{i=1}^n x_i(t) = b$ for $t\geq 0$.

We provide an additional simulation using a logarithmic penalty term in the form $\frac{1}{\mu}(\log (1+\exp(\mu (x_i-M_i)))+\log (1+\exp(\mu (m_i-x_i))))$ to address the box constraints. For this penalty term, we set $\mu=5$, which is sufficiently large to properly enforce the box constraints. All other simulation parameters remain the same as in the previous case. The results are presented in Fig.~\ref{fig_dyn2}.

\subsection{Resilience to Link Failure}
\begin{figure*}
\centering
\subfigure[$p_l=0$]
{       \includegraphics[width=1.7in]{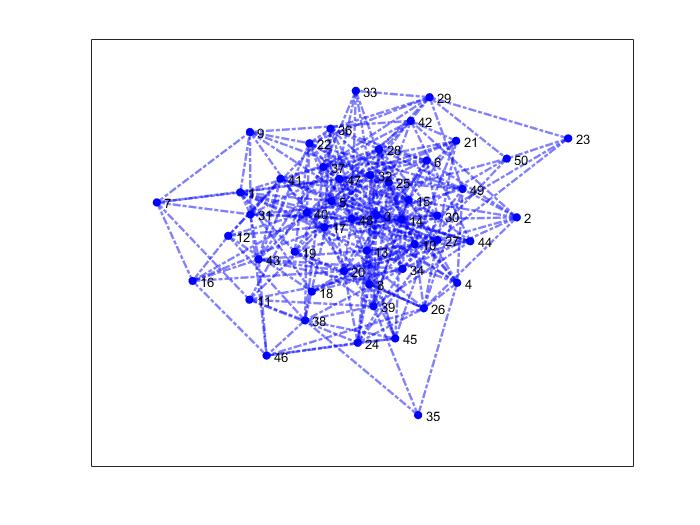}}
\subfigure[$p_l=0.5$]
{ 		\includegraphics[width=1.7in]{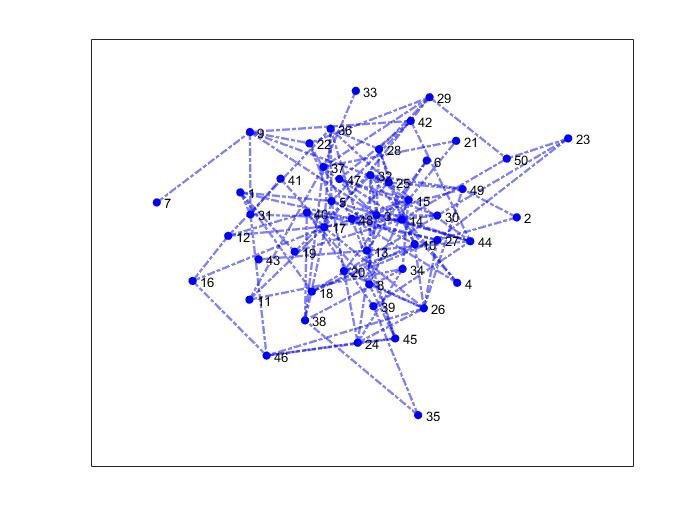}}
\subfigure[$p_l=0.7$]
{  		\includegraphics[width=1.7in]{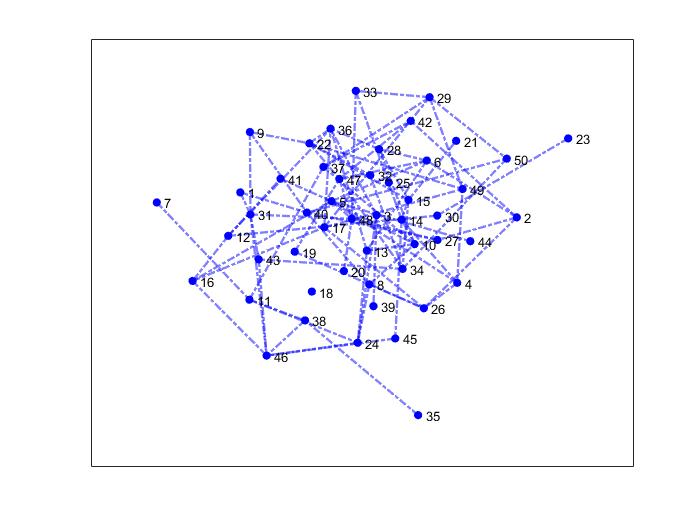}}
\subfigure[$p_l=0.85$]
{  		\includegraphics[width=1.7in]{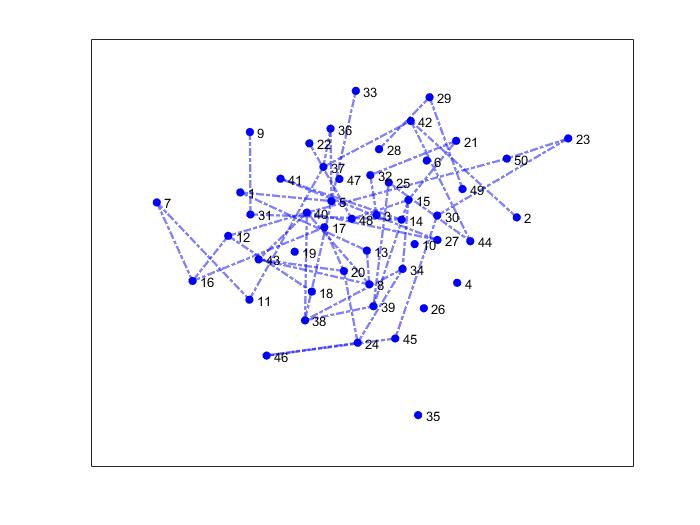}}
\subfigure[$p_l=0.92$]
{ 		\includegraphics[width=1.7in]{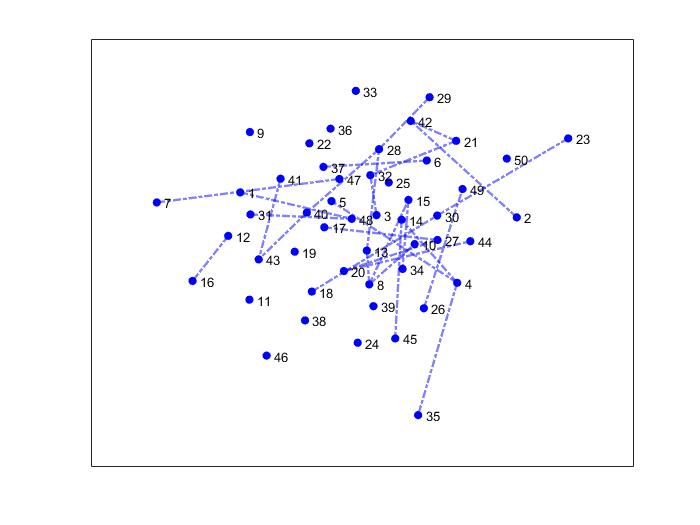}}
\subfigure[$p_l=0.85$, $T=2$]
{ 		\includegraphics[width=1.7in]{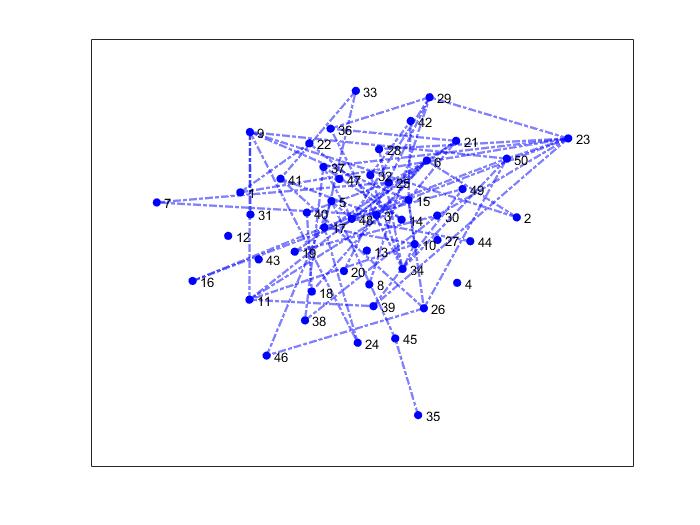}}
\subfigure[$p_l=0.92$, $T=4$]
{ 		\includegraphics[width=1.7in]{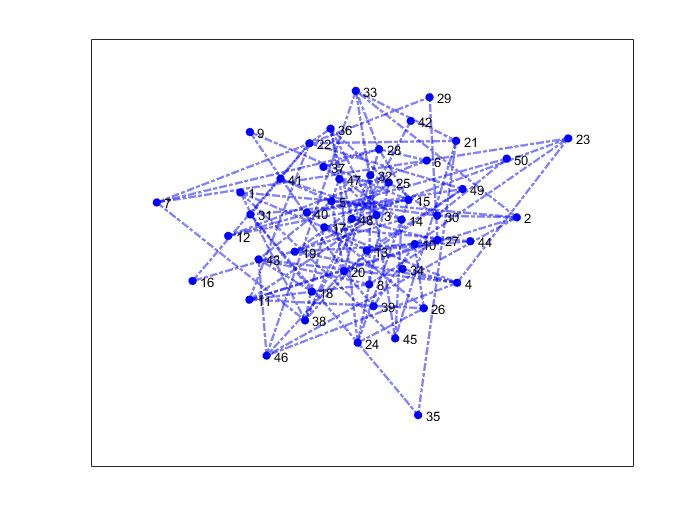}}
\caption{Network topology under varying link failure conditions. (a) Original Erdős--Rényi network with 20\% linking probability. (b)-(c) Networks with failure rates below percolation threshold ($p_l < p_c$), maintaining connectivity with probability 1. (d)-(e) Networks with failure rates exceeding percolation threshold ($p_l > p_c$), becoming disconnected. (f)-(g) Union networks over $T=2$ and $T=4$ time steps respectively, demonstrating uniform connectivity when $p_l^{T+1} < p_c$ despite instantaneous disconnections.}
\label{fig_graphs}
\end{figure*}

Now, we examine an Erdős--Rényi (ER) network with linking probability $p=20\%$ subjected to various link failure rates (i.e., $p_l\in\{50\%$, $70\%$, $85\%$, $92\%\}$) to systematically evaluate network resilience. The bond-percolation threshold for this graph configuration is $p_c=79.5\%$, as determined by~\eqref{eq_pc}. This threshold creates two distinct operational regimes: for $p_l<p_c$, network connectivity is preserved with probability approaching 1 despite random link failures; for $p_l>p_c$, the network loses instantaneous connectivity but maintains uniform connectivity over extended time windows $T$ that satisfy the conditions in Theorem~\ref{thm_drop} and~\eqref{eq_B*}.

Specifically, when $p_l=85\%$ and $p_l=92\%$, uniform connectivity is maintained over $T=2$ and $T=4$ time steps, respectively. Figure~\ref{fig_graphs} illustrates representative network topologies under these different failure probabilities, demonstrating the structural changes that occur as failure rates increase.

To simulate realistic packet drops, our model implements dynamic link failures at each time step. Links randomly disconnect (simulating packet drops) and later reconnect (representing successful packet delivery), while other links may simultaneously fail with probability $p_l$. This creates a highly dynamic network topology with continuously changing connections between agents. As expected, higher link failure rates result in slower convergence of our resilient dynamics due to reduced information exchange across the network, as demonstrated in Figure~\ref{fig_fail} with step rate $\eta=0.2$.

\begin{figure}[]
\centering
\includegraphics[width=2.5in]{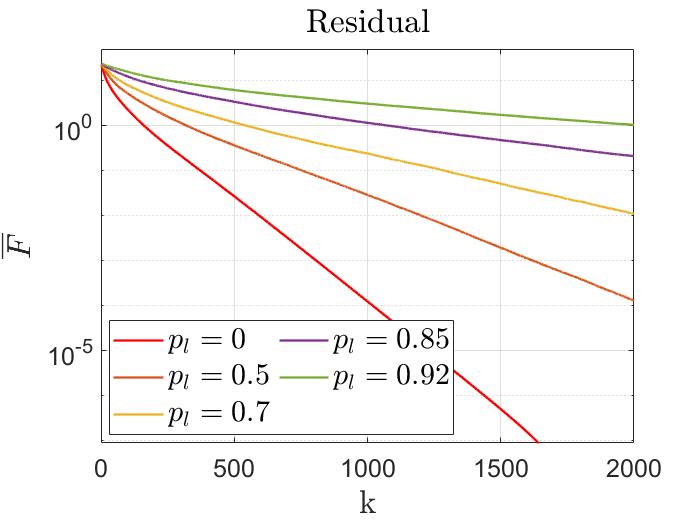}
\caption{Impact of link failure probability on convergence performance. The plot shows cost residual evolution under varying link failure rates (50\%-92\%), demonstrating that higher failure probabilities require longer time windows for uniform connectivity, resulting in slower but still guaranteed convergence to optimal solutions.}
\label{fig_fail}
\end{figure}

\subsection{Resilience to Time-Delay}
Next, we evaluate the resilience of our distributed resource allocation algorithm to communication delays by implementing the time-delayed dynamics in Equation~\eqref{eq_sol_delay}. Our testbed uses the same ER network configuration (with linking probability $p=20\%$) while introducing heterogeneous time-varying delays across links. We conduct experiments with maximum delay bounds $\overline{\tau} \in \{2,4,6\}$ time steps to analyze performance degradation under increasing latency conditions.

For this network topology, spectral analysis yields $\lambda_2=0.044$ and $\lambda_n=0.311$, with sector-bound parameters $\kappa_n=1$, $\kappa_l=0.938$, $\mathcal{K}_n=1$, $\mathcal{K}_l=1.062$, and smoothness parameter $u\approx 0.115$. Substituting these values into the theoretical bound provided by~\eqref{eq_tau_eta}, we determine that with step-rate $\eta=2$, convergence is guaranteed only for $\overline{\tau} \leq 2$.

This theoretical prediction is validated by our experimental results shown in Figure~\ref{fig_delay}(left), where convergence is observed for $\overline{\tau}=2$ but not for $\overline{\tau}=4$ or $\overline{\tau}=6$. Following the relationship established in~\eqref{eq_tau_eta} and discussed in Remark~\ref{rem_delay}, we can accommodate larger time delays by reducing the step-rate. When decreasing $\eta$ to $0.5$, as illustrated in Figure~\ref{fig_delay}(right), the algorithm successfully converges for all tested delay bounds $\overline{\tau} \in \{2,4,6\}$, though at a correspondingly slower convergence rate -- demonstrating the fundamental trade-off between delay tolerance and convergence speed.

\begin{figure}[]
\centering
\includegraphics[width=2.5in]{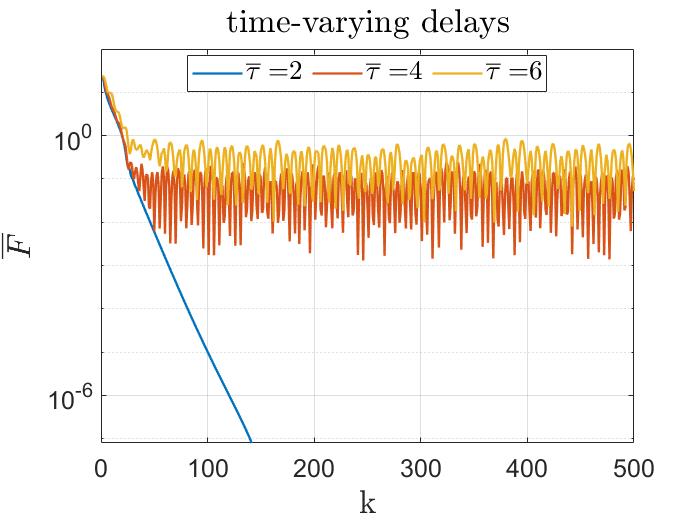}
\includegraphics[width=2.5in]{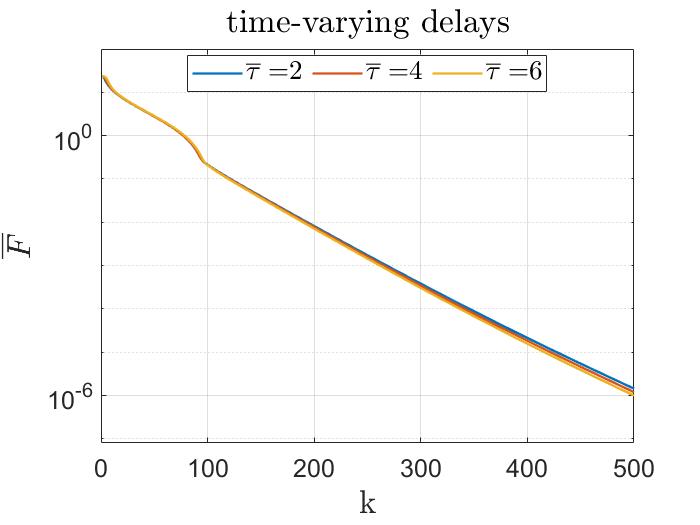}
\caption{Algorithm resilience to communication delays. Left: Convergence behaviour with step-rate $\eta=2$, showing successful convergence for delays $\overline{\tau}=2$ but divergence for larger delays. Right: Improved resilience with reduced step-rate $\eta=0.5$, achieving convergence for all tested delay bounds ($\overline{\tau}\in\{2,4,6\}$), illustrating the fundamental trade-off between delay tolerance and convergence speed.}
\label{fig_delay}
\end{figure}

\subsection{Practical Engineering Example and Comparison with the Literature}

As a practical engineering application, we demonstrate our algorithm on the distributed economic dispatch problem in power systems, where generation resources must be optimally allocated across a network of interconnected generators. This problem, extensively studied in \cite{cherukuri2016initialization, kar2012distributed}, requires distributing a total power demand of $b = 600$ MW among $n = 10$ generators while minimizing generation costs.

Each generator $i$ has a quadratic cost function that models its operational expenses described by
\begin{align}
f_i(p_i) = a_i p_i^2 + b_i p_i + c_i,
\end{align}
where $p_i$ represents the power output (in MW) of generator $i$, and coefficients $a_i$, $b_i$, and $c_i$ characterize the specific cost profile of each generator. These parameters are configured according to realistic generator data provided in~\cite{kar2012distributed}.

Practical operating constraints require that each generator's output remains within its capacity limits, represented as $m_i \leq p_i \leq M_i$, where $m_i = 20$ MW and $M_i = 110$ MW denote the minimum and maximum generation capacities, respectively. We enforce these constraints through penalty terms of the form $\sigma([p_i - M_i]^+ + [m_i - p_i]^+)$, where $[u]^+ = (\max\{0, u\})^2$ and $\sigma=40$ is a sufficiently large penalty coefficient.

We implement Algorithm~\ref{alg_ac} with both log-scale quantization for data transmission and sign-based nonlinearity to enhance convergence performance. Figure~\ref{fig_edp} illustrates the evolution of both the cost residual and the power allocations over time. For comparative analysis, we benchmark our method's convergence rate against several state-of-the-art approaches: standard linear methods \cite{cherukuri2016initialization, kar2012distributed}, acceleration-based techniques \cite{shames2011accelerated}, finite-time convergent algorithms \cite{chen2016distributed}, and fixed-time convergent solutions~\cite{ning2017distributed}.
For all these allocation techniques, the network topology is considered as an ER network with linking probability of $p=20\%$. As a comparison, although the approach proposed by \cite{shames2011accelerated} reaches faster convergence using momentum-based linear dynamics, it cannot address nonlinear mappings (e.g., due to log-scale quantization and saturation), cannot handle communication time-delays over the network, and only converges over fixed and time-invariant networks. On the other hand, our proposed approach handles nonlinear mapping and is resilient to communication time-delays and link failues by addressing convergence over uniformly-connected time-varying networks. As compared to other solutions, the convergence of our proposed solution is comparable and even better in some cases. It is worth mentioning that the convergence rate of our proposed algorithm can be improved by using nonlinear signum-based mappings at the node dynamics. 

\begin{figure}[]
	\centering
	\includegraphics[width=2.5in]{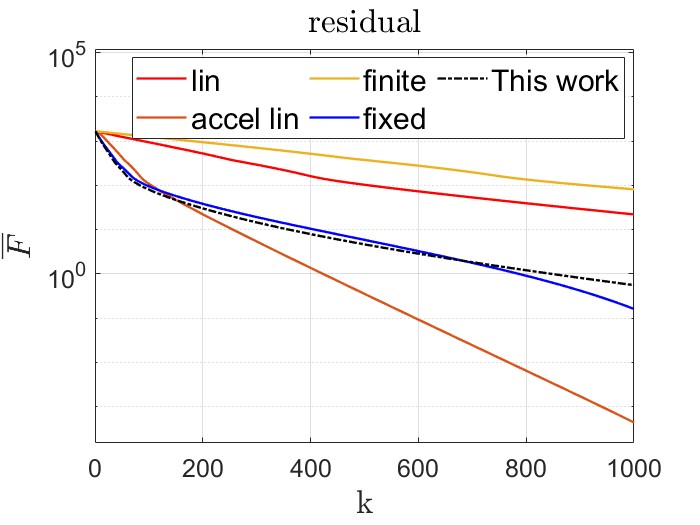}
	\includegraphics[width=2.5in]{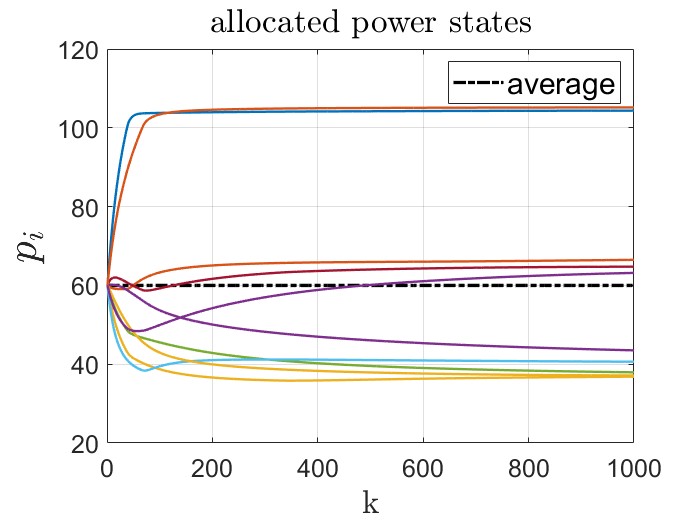}
\caption{Economic dispatch performance comparison. Left: Convergence rates of our proposed algorithm versus existing approaches (linear, accelerated linear, finite-time, and fixed-time methods) for power allocation optimization. Right: Evolution of individual generator outputs, maintaining constant sum throughout the iteration process, demonstrating strict adherence to power demand constraints at all times.}
	\label{fig_edp}
\end{figure}

To evaluate our algorithm's resilience to communication imperfections, we conduct extensive simulations under two critical network disruption scenarios: (i) stochastic link failures and (ii) heterogeneous communication delays.

For the link failure analysis, we examine the algorithm's performance under progressively deteriorating network conditions with failure probabilities $p_l \in \{50\%, 70\%, 85\%, 92\%\}$. This range strategically spans both sub-critical ($p_l < p_c$) and super-critical ($p_l > p_c$) regimes relative to the network's percolation threshold $p_c$.

Similarly, for communication delay analysis, we test our algorithm's robustness under heterogeneous time delays bounded by maximum values $\overline{\tau} \in \{3, 7, 10\}$ time steps. These delays are implemented as random, time-varying quantities that differ across communication links, accurately modelling real-world network conditions.

Figure~\ref{fig_edp2} presents the convergence behaviour under these challenging scenarios, demonstrating that our algorithm maintains convergence despite significant communication impairments. This robust performance stands in stark contrast to existing approaches in the literature \cite{cherukuri2016initialization, kar2012distributed, shames2011accelerated, chen2016distributed, ning2017distributed}, none of which provide theoretical guarantees for convergence under both time-delayed communications and uniformly-connected network topologies.

The comparative advantage of our approach is particularly evident in practical deployment scenarios where packet drops and communication delays are inevitable. While conventional algorithms fail to maintain convergence in such environments, our solution demonstrates reliable performance across varying degrees of link failures and communication delays, confirming both the theoretical resilience guarantees established in Theorems~\ref{thm_drop} and \ref{thm:delay_convergence} and the practical utility of our approach for real-world distributed systems.

\begin{figure}[]
\centering
\includegraphics[width=2.5in]{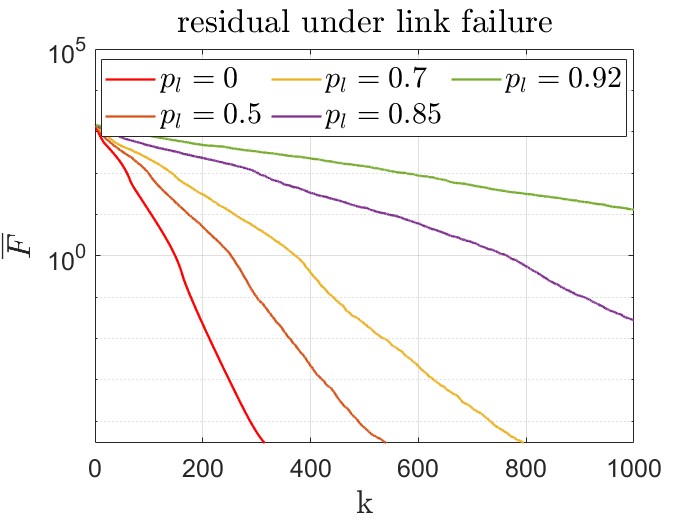}
\includegraphics[width=2.5in]{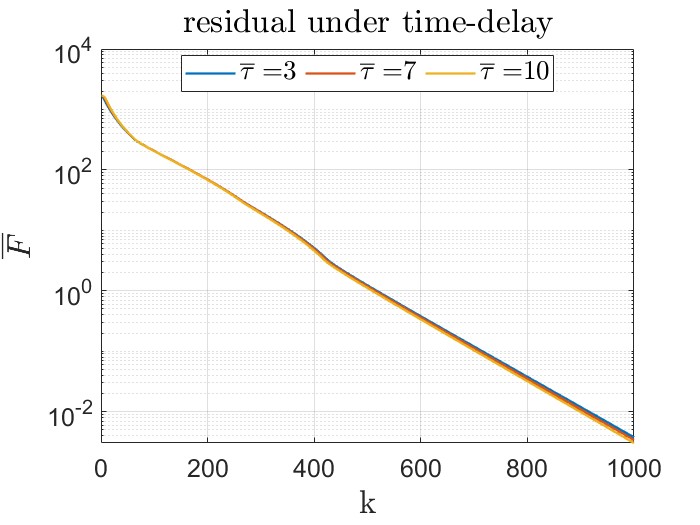}
\caption{Resilience of economic dispatch under communication imperfections. Left: Power allocation convergence under various link failure rates, demonstrating robust performance even with severe network disruptions. Right: Convergence behaviour with heterogeneous time delays, showing the algorithm's ability to maintain optimality despite communication latency in practical power systems.}
\label{fig_edp2}
\end{figure}

\section{Conclusions} \label{sec_con}
This paper introduces a novel framework for resilient distributed resource allocation that fundamentally advances the state-of-the-art in networked multi-agent systems operating under challenging communication conditions. Our approach makes three significant theoretical and practical contributions:

First, we have developed a distributed algorithm that maintains all-time feasibility for resource-demand balance -- a critical property ensuring system stability even during network disruptions or early termination. Unlike existing approaches that provide only asymptotic feasibility, our solution guarantees constraint satisfaction at every iteration.

Second, our algorithm achieves provable convergence to optimal solutions despite random link failures and heterogeneous communication delays. By leveraging percolation theory and spectral graph analysis, we established precise mathematical bounds on tolerable failure rates and delay thresholds, providing system designers with concrete guidelines for implementation in real-world networks.

Third, extensive simulations across various network configurations -- including dynamic topologies, severe link failures (up to $92\%$), and substantial communication delays -- demonstrate the algorithm's remarkable resilience in practical scenarios such as economic dispatch in power systems.

One direction of future research is to extend the results to consider DRA over directed network topologies.
The theoretical framework developed in this paper extends beyond resource allocation, providing a foundation for enhancing resilience in various distributed systems. Future research can build upon our approach to address emerging challenges in distributed machine learning (particularly federated learning under unreliable communication), multi-agent reinforcement learning, and edge computing applications where network reliability cannot be guaranteed. Additionally, our methodology could be adapted to handle adversarial attacks and strategic network disruptions, further advancing the robustness of distributed systems in security-critical applications.

\section*{Acknowledgement}
This work is funded by Semnan University, grant No. 226/1403/1403214.

\bibliographystyle{elsarticle-num-names}
\bibliography{bibliography}

\end{document}